%% file: diversity.tex
\documentclass{sig-alternate}

\usepackage{times}
\usepackage{graphicx}
\usepackage{amsmath, amsfonts}
\usepackage{color}
\usepackage{url}
\usepackage{cite}

\newtheorem{definition}{Definition}
\newtheorem{problem}{Problem}
\newtheorem{theorem}{Theorem}
\newtheorem{corollary}{Corollary}
\newtheorem{lemma}{Lemma}
\newtheorem{property}{Property}

\newcommand{\nop}[1]{}

\newcommand{\T}{\mathcal{T}}
\newcommand{\Qo}{\dot{Q}}
\newcommand{\Ro}{\dot{R}}
\newcommand{\Qt}{\ddot{Q}}
\newcommand{\Rt}{\ddot{R}}
\newcommand{\A}{\mathcal{A}}
\newcommand{\C}{\mathcal{C}}

\def\done{\hspace*{\fill} $\square$}

\def\extraspacing{\vspace{2mm}}

\def\figcapup{\vspace{-2mm}}
\def\figcapdown{\vspace{-4mm}}
\def\tblcapup{\vspace{-2mm}}
\def\tblcapdown{\vspace{-4mm}}

\begin{document}
\begin{sloppy}

\title{The Hardness and Approximation Algorithms for L-Diversity}

\numberofauthors{3}

\author{
\alignauthor
Xiaokui Xiao\\
       \affaddr{Nanyang Technological University}\\
       \affaddr{Singapore}\\
       \email{xkxiao@ntu.edu.sg}
\alignauthor
Ke Yi\\
       \affaddr{Hong Kong University of Science and Technology}\\
       \affaddr{Hong Kong}\\
       \email{yike@cse.ust.hk}
\alignauthor
Yufei Tao\\
       \affaddr{Chinese University of Hong Kong}\\
       \affaddr{Hong Kong}\\
       \email{taoyf@cse.cuhk.hk}
}

\toappear{Permission to make digital or hard copies of all or part of
  this work for personal or classroom use is granted without fee provided that
  copies are not made or distributed for profit or commercial advantage and
  that copies bear this notice and the full citation on the first page. To copy
  otherwise, to republish, to post on servers or to redistribute to lists,
  requires prior specific permission and/or a fee. \par
  {\confname EDBT 2010}, March 22--26, 2010, Lausanne, Switzerland.\par
  Copyright 2010 ACM 978-1-60558-945-9/10/0003\ ...\$10.00}

\maketitle

\begin{abstract}
The existing solutions to privacy preserving publication can be
classified into the {\em theoretical} and {\em heuristic}
categories. The former guarantees provably low information loss,
whereas the latter incurs gigantic loss in the worst case, but is
shown empirically to perform well on many real inputs. While
numerous heuristic algorithms have been developed to satisfy
advanced privacy principles such as $l$-diversity, $t$-closeness,
etc., the theoretical category is currently limited to $k$-anonymity
which is the earliest principle known to have severe vulnerability
to privacy attacks. Motivated by this, we present the first
theoretical study on $l$-diversity, a popular principle that is
widely adopted in the literature. First, we show that optimal
$l$-diverse generalization is NP-hard even when there are only 3
distinct sensitive values in the microdata. Then, an $(l \cdot
d)$-approximation algorithm is developed, where $d$ is the
dimensionality of the underlying dataset. This is the first known
algorithm with a non-trivial bound on information loss. Extensive
experiments with real datasets validate the effectiveness and efficiency
of proposed solution.

\end{abstract}

\section{Introduction}\label{sec:intro}

Privacy preserving publication has become an active topic in
databases. An important problem is the prevention of {\em linking
attacks} \cite{s01,s02-b}. To explain this threat, assume that a
hospital releases the patients' details in
Table~\ref{tbl:intro-micro}, called the {\em microdata}, to medical
researchers. {\em Disease} is a {\em sensitive attribute} (SA)
because a patient's disease is regarded as her/his privacy.
Attribute {\em Name} is {\em not} part of the table, but it will be
used to facilitate tuple referencing. Consider an adversary that
knows (i) the age ($< 30$), gender (M) and education level
(bachelor) of Calvin, and (ii) Calvin has a record in the microdata.
Thus, s/he easily finds out that Tuple 3 is Calvin's record and
hence, Calvin contracted pneumonia.

In the above attack, columns {\em Age}, {\em Gender}, and {\em
Education} are {\em quasi-identifier} (QI) attributes because they
can be combined to reveal an individual's identity. The cause of
privacy leakage is that an individual (e.g., Calvin) may have a
unique set of QI values. A common approach for fixing this problem
is {\em generalization}, which partitions the microdata into {\em
QI-groups}, and then, converts the QI values in each group to the
same form, e.g., replaces distinct values on each QI attribute with stars. For example, Table~\ref{tbl:intro-micro-2-anony} shows a
generalization of Table~\ref{tbl:intro-micro}, based on a partition
of four QI-groups. Notice that, in the second QI-group, the ages of
Tuples 3 and 4 have been {\em suppressed} into stars, since their
original values are different.

\begin{table}[t]
\includegraphics[height=36mm]{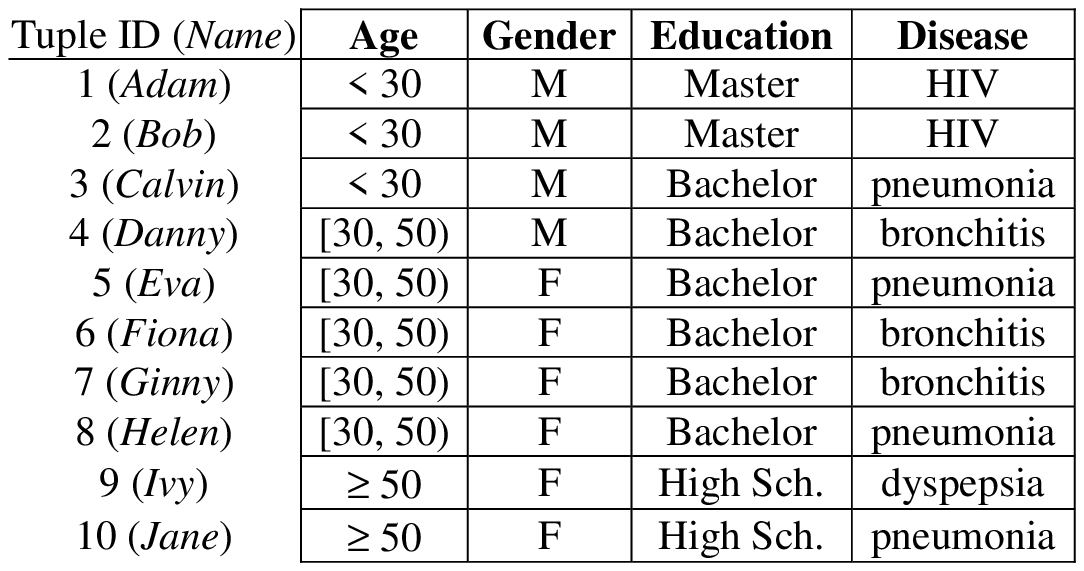}
\tblcapup \caption{The microdata} 
\label{tbl:intro-micro}
\end{table}

\begin{table}[t]
\hspace{2mm}
\includegraphics[height=36mm]{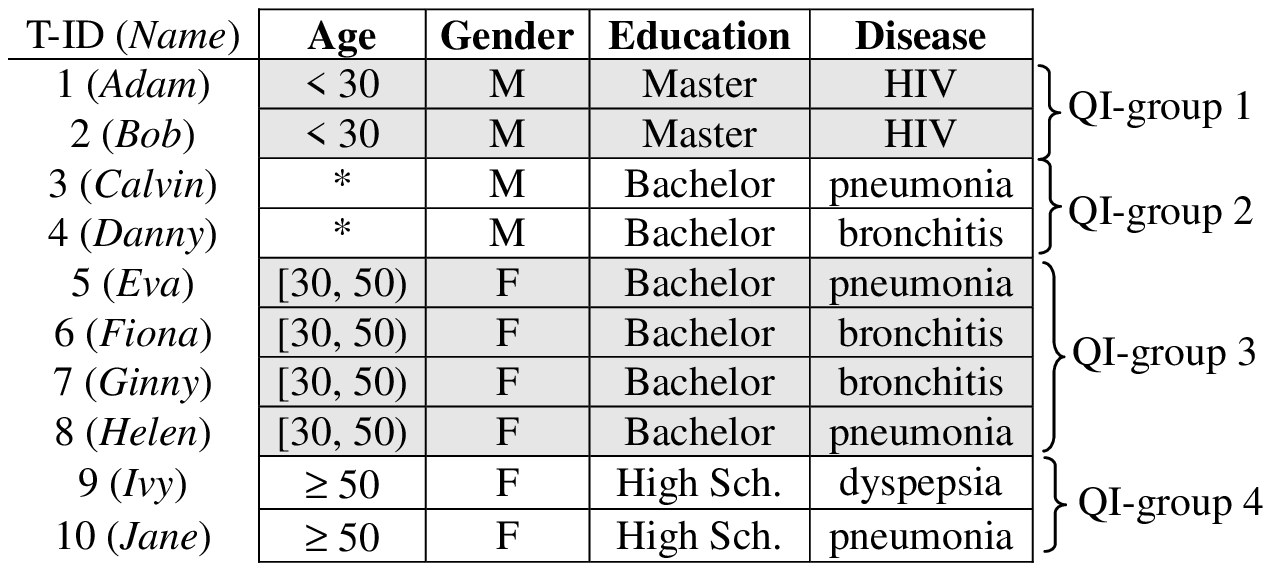}
\tblcapup  \caption{2-anonymous publication} \tblcapdown
\label{tbl:intro-micro-2-anony}
\end{table}

A generalized table can be released if it satisfies an {\em
anonymization principle}, which determines the quality of privacy
protection. The earliest principle is {\em $k$-anonymity} \cite{s01,s02-b}, which requires each QI-group to contain at least $k$ tuples.
As a result, each tuple carries the same QI values as at least $k -
1$ other tuples. For instance, Table~\ref{tbl:intro-micro-2-anony}
is 2-anonymous. Given this table, the adversary mentioned earlier
cannot tell whether Tuple 3 or 4 belongs to Calvin.

Machanavajjhala et al.\ \cite{mgk+07} observe that $k$-anonymity
suffers from the {\em homogeneity problem}: a QI-group may have too
many tuples with the same SA (sensitive attribute) value. For
example, both tuples in the first QI-group of
Table~\ref{tbl:intro-micro-2-anony} have HIV. As a result, an
adversary having the QI particulars of Adam (or Bob) can assert that
Adam (Bob) has HIV, without having to identify the tuple owned by
Adam (Bob). Note that the problem cannot be eliminated by increasing
$k$, because $k$-anonymity places no constraint on the SA values in
each QI-group.

The above problem has led to the development of numerous {\em
SA-aware} principles, which set forth conditions to be fulfilled by
the SA values in each QI-group. Among the existing principles, {\em
$l$-diversity} \cite{mgk+07} is the most widely deployed
\cite{gkkm07,kg06,mgk+07,wlfw06,xt06-b,xt07}, due to its
simplicity and good privacy guarantee. Specifically, this principle
demands\footnote{Precisely speaking, $l$-diversity requires each
QI-group to have at least $l$ well-represented values. There are
different interpretations of ``well-represented" \cite{mgk+07}. The
version discussed here is widely adopted in the literature
\cite{gkkm07,wfwp07,xt06-b}.} that, in each QI-group, at most $1 /
l$ of its tuples can have an identical SA value.
Table~\ref{tbl:intro-micro-2-div} demonstrates a 2-diverse
generalization of Table~\ref{tbl:intro-micro}. It can be easily
verified that, in each QI-group, the frequency of each SA value is
at most 50\%. Thus, even if an adversary figures out the QI-group
containing the record of an individual, s/he can determine the real
SA value of the individual with no more than 50\% confidence.

\begin{table}[t]
\centering
\includegraphics[height=36mm]{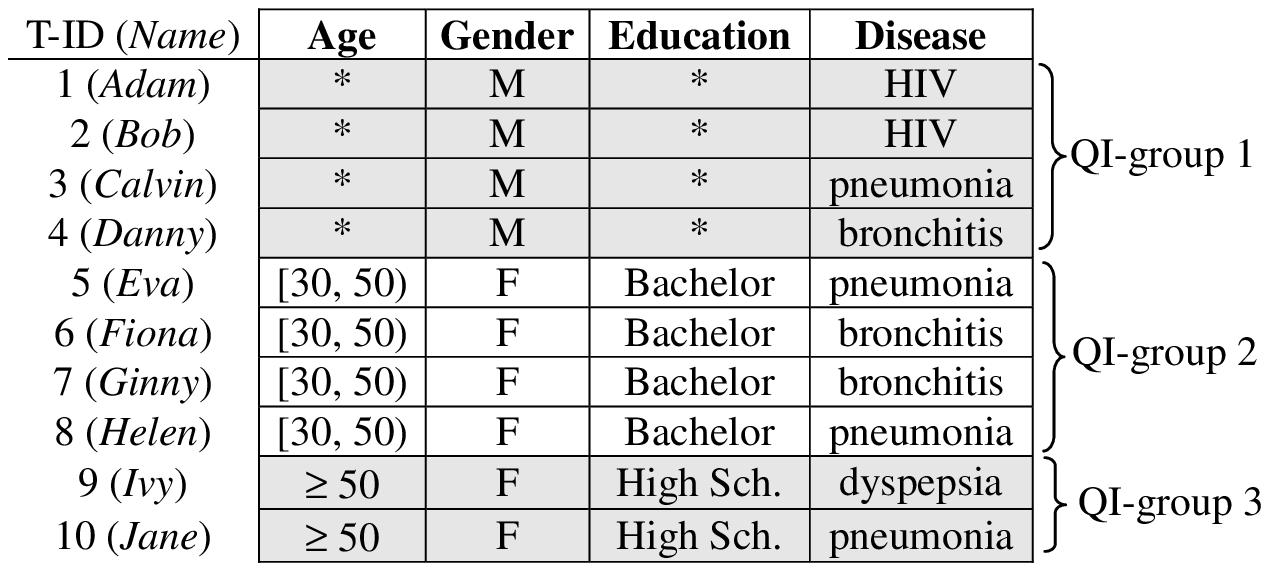}
\tblcapup  \caption{2-diverse publication} \tblcapdown
\label{tbl:intro-micro-2-div}
\end{table}

\subsection{Theory Stops at {\large $k$}-anonymity} \label{sec:intro-motivate}

The goal of privacy preserving publication is to minimize the
information loss (e.g., the number of stars used in
Tables~\ref{tbl:intro-micro-2-anony} and
\ref{tbl:intro-micro-2-div}) in enforcing the selected anonymization
principle. The existing solution can be divided into two categories:
{\em theoretical} and {\em heuristic}. The former develops
algorithms with {\em worst-case} performance bounds. The
latter, on the other hand, designs algorithms that work well on many
real datasets, but may have very poor performance (i.e., incur
gigantic information loss) on ``unfriendly" inputs.

We notice that, in terms of privacy protection, the theoretical
category significantly lags behind its heuristic counterpart. As
reviewed in Section~\ref{sec:related}, many heuristic algorithms
exist for various SA-aware principles that ensure strong privacy
preservation. However, as surveyed next, all the theoretical results
concern with only $k$-anonymity, and none of them deals with SA-aware
principles. In other words, currently all the theoretical algorithms
suffer from the homogeneity problem mentioned earlier, and thus, are
weak in privacy guarantees. Note that, this drawback also reduces
the practical usefulness of their nice bounds of information loss,
because a publisher puts privacy at a higher priority than utility.

In the theoretical category, Meyerson and Williams \cite{mw04} are
the first to establish the complexity of optimal $k$-anonymity, by showing that it is NP-hard to compute a $k$-anonymous table that contains the minimum number of stars (i.e., suppressed values). They also
provide a $O(k \log k)$-approximation algorithm. Aggarwal et al.\
\cite{afk05} offer a stronger NP-hardness proof that requires a
smaller domain of the QI attributes. They also improve the
approximation ratio to $O(k)$. Park and Shim \cite{ps07} enhance the
ratio further to $O(\log k)$. It should be noted that, the
algorithms in \cite{mw04,ps07} have running time exponential in
$k$, while the running time of the algorithm in \cite{afk05} is a polynomial
of $k$ and $n$. Du et al.\ \cite{dxt+07} consider the case when the generalized table is produced not by replacing QI values with stars, but by applying {\em multi-dimensional generalization} (see Section~\ref{sec:related}). They show that enforcing $k$-anonymity in this setting is still NP-hard, and give an
$O(d)$ approximation algorithm, where $d$ is the number of QI
attributes. Aggarwal et al.\ \cite{afk+06} propose {\em
clustering-based generalization}, prove the NP-hardness of
$k$-anonymity (in \cite{afk+06}, $k$ is replaced by $r$), and
provide constant approximation solutions.

\subsection{Our Results} \label{sec:intro-cont}

This paper presents the first theoretical study on $l$-diverse
anonymization. In particular, we consider that the microdata is anonymized by suppressing QI values, and we aim at achieving $l$-diversity with the minimum number of stars. At first glance, a simple reduction from
$k$-anonymity seems to establish the NP-hardness of $l$-diversity.
Specifically, given a table where no two tuples have the same SA
value, the optimal $l$-diversity generalization is also the optimal
``$l$-anonymity" generalization. Hence, if there was an optimal
$l$-diverse algorithm that runs in polynomial time, the same
algorithm can efficiently solve optimal $k$-anonymity as well, which contradicts the NP-hardness of optimal $k$-anonymity.

The previous reduction requires that the number $m$ of distinct SA
values is as large as the cardinality $n$ of the microdata.
Therefore, a natural question is whether optimal $l$-diversity can
be settled in polynomial time if $m \ll n$, as is true in practice.
In fact, the answer is apparently ``yes" for $m = 2$, in which case
the problem becomes bipartite matching (see
Section~\ref{sec:hardness}), a well-known polynomial-time solvable
problem. The first major contribution of our work is a proof showing that
$l$-diversity is NP-hard as long as $m \ge 3$. Clearly, this result
is much stronger than the hardness result from the earlier simple
reduction. In fact, our result still holds even if the alphabet
(i.e., the domain union of all attributes) has a size of only $m +
1$.

On the algorithm side, we propose a solution that ensures an approximation ratio of $l \cdot d$, where $d$ is the the number of QI attributes in the microdata. This is the first algorithm on $l$-diversity with a non-trivial bound of information loss. Furthermore, our algorithm is also highly efficient -- it runs in close-to-linear time. Although the $l \cdot d$ approximation ratio may trigger concerns about the usefulness of our technique in practice, we note that the actual performance of our algorithm is much better than the theoretical bounds. Specifically, our algorithm executes in three phases, and depending on the dataset characteristics, may finish in any phase. Termination in the first one results in a $d$-approximate solution, while termination at the second phase incurs at most $l \cdot d$ additional stars (with respect to the $d$-approximation). In any case, an $(l \cdot d)$-approximate solution is guaranteed after the third phase. On the large set of datasets tested in our experiments, our algorithm always terminates before the third phase, thus achieving an approximation ratio of $d$. In addition, our algorithm can be further improved, when we combine it with a heuristic-based $l$-diversity solution. Empirical evaluation shows that, such a hybrid method significantly outperforms the existing $l$-diversity algorithms, in terms of the number of stars required in anonymization.

The rest of the paper is organized as follows.
Section~\ref{sec:related} surveys the previous work relevant to
ours. Section~\ref{sec:prob} formally defines the problem. Section~\ref{sec:hardness}
establishes the hardness of the problem, while
Section~\ref{sec:tmin} presents our approximation algorithm and proves
its quality guarantees. Section~\ref{sec:exp} experimentally
evaluates the effectiveness and efficiency of the proposed
technique. Finally, Section~\ref{sec:conclude} concludes the paper
with directions for future work.

\section{Related Work} \label{sec:related}

The existing theoretical results on privacy preserving publication
have been explained in Section~\ref{sec:intro-motivate}. In the
following, we focus on other approaches based on four categories.

\extraspacing
\noindent {\bf Anonymization methodologies } Most of the existing work on microdata publication adopts generalization to anonymize data. There exist three variations of generalization, namely, {\em suppression} \cite{aw89}, {\em single-dimensional generalization} \cite{i02,ba05,fwy05,xwp+06,wlfw06}, and {\em multi-dimensional generalization} \cite{ldr06-a,ldr06-b,gkkm07}. Suppression replaces distinct QI values in each QI-group with stars, as demonstrated in Section~\ref{sec:intro}. Single-dimensional generalization, on the other hand, divides the domain of each QI attribute into {\em disjoint} sub-domains, and maps each QI value in the microdata to the sub-domain that contains the value, i.e, it ``coarsens'' the domains of the QI attributes. For example, Table~\ref{tbl:related-single-dimen} illustrates a single-dimensional generalization of Table~\ref{tbl:intro-micro} that satisfies $2$-diversity. In particular, the domain of {\em Age} ({\em Education}) is divided into two sub-domains, ``$<$50'' and ``$\ge$50'' (``High school or below'' and ``Bachelor or above''). Multi-dimensional generalization is an extension of single-dimensional generalization: it allows QI values to be mapped to overlapping sub-domains. For instance, Table~\ref{tbl:related-multi-dimen} shows a $2$-diverse multi-dimensional generalization of Table~\ref{tbl:intro-micro}.

\begin{table}[t]
\centering
\includegraphics[height=32mm]{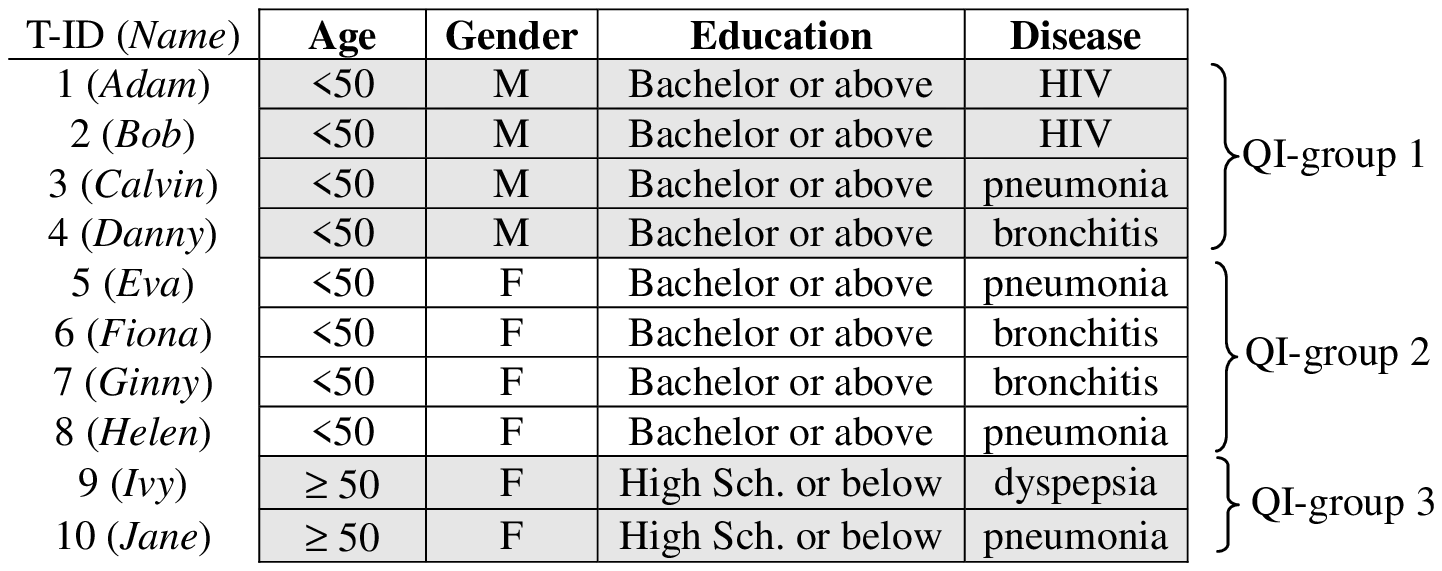}
\tblcapup  \caption{Single-dimensional generalization} 
\label{tbl:related-single-dimen}
\end{table}

\begin{table}[t]
\centering
\includegraphics[height=33mm]{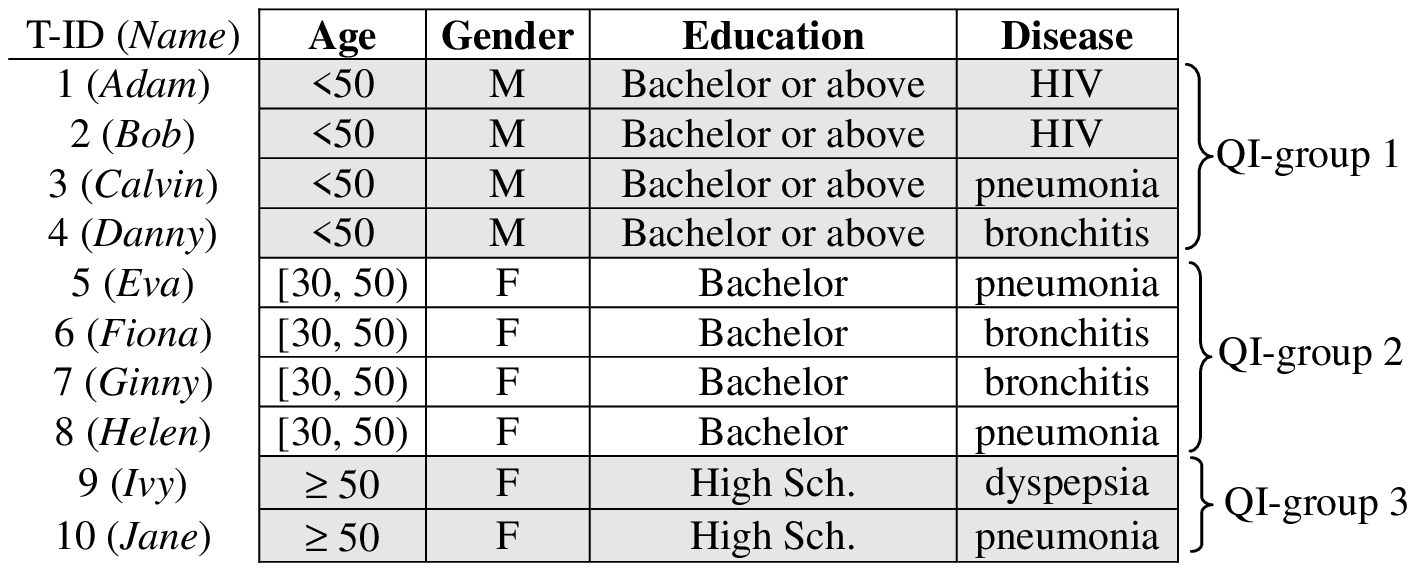}
\tblcapup  \caption{Multi-dimensional generalization} \tblcapdown
\label{tbl:related-multi-dimen}
\end{table}

As multi-dimensional generalization imposes fewer constrains on how the QI values should be transformed, it can retain more information in the anonymized tables than suppression and single-dimensional generalization. For example, it can be verified that each QI value in Table~\ref{tbl:related-multi-dimen} is equally or more accurate than the corresponding value in Table \ref{tbl:intro-micro-2-div} or \ref{tbl:related-single-dimen}. However, suppression and single-dimensional generalization have a significant advantage over multi-dimensional generalization: the anonymized data they produce that can be directly used by off-the-shelf softwares (e.g., SAS \cite{sas}, SPSS \cite{spss}, Stata \cite{stata}) designed for microdata analysis. Specifically, tables with suppressed values can be processed as microdata with missing entries. On the other hand, any single-dimensional generalization can be treated as a microdata table defined over attributes with coarsened domains, i.e., all analysis on the data are performed by regarding each sub-domain of a QI attribute as a unit value.

In contrast, multi-dimensional generalization results in data that cannot be handled by existing softwares, due to the complex relationships among QI values represented by overlapping sub-domains. To understand this, consider that a user wants to count the number of individuals in Table~\ref{tbl:related-multi-dimen} with ages in [30, 50). In that case, the user has to take into account not only the tuples with an {\em Age} value [30, 50), but also those with a value ``$<$50'', which is non-trivial since it is difficult to decide how those tuples may contribute to the query result. In general, performing analysis (e.g., regression, classification) on overlapping sub-domains is highly complicated, and hence, is not supported by off-the-shelf statistical softwares. This explains why existing anonymization systems, like $\mu$-Argus \cite{hw96} and Datafly \cite{s97}, adopt suppression and single-dimensional generalization instead of multi-dimensional generalization.

In summary, suppression and single-dimensional generalization are more preferable, if the data publisher aims to release data that can be easily used by ordinary users; otherwise, multi-dimensional generalization can be adopted. An interesting question is, how does suppression compare with single-dimensional generalization? To answer this question, in Section~\ref{sec:exp-kl}, we will experimentally evaluate our suppression algorithms against the existing single-dimensional generalization methods.

Besides generalization, there also exists other methodologies for privacy preserving data publication. Kifer and Gehrke
\cite{kg06} propose {\em marginal publication}, which releases
different projections of the microdata onto various sets of
attributes. Xiao and Tao \cite{xt06-b} advocate {\em anatomy} that
publishes QI and SA values directly in separate tables. Aggarwal and
Yu \cite{as04} design the {\em condensation method}, which releases
only selected statistics about each QI-group. Rastogi et al.\
\cite{rhs07} employ the {\em perturbation} approach.

\extraspacing
\noindent {\bf Anonymization principles } Privacy protection must take into
account the knowledge of adversaries. A common assumption is that an
adversary has the precise QI values of all individuals in the
microdata. Indeed, these values can be obtained, for example, by
knowing a person or consulting an external source such as a voter
registration list \cite{s02-b}.

Under this assumption, both $k$-anonymity and $l$-diversity aim at
preventing the accurate inference of individuals' SA values. Many
other principles share this objective. {\em $(\alpha, k)$-anonymity}
\cite{wlfw06} combines the previous two principles: each QI-group
must have size $k$ and at most $\alpha$ percent of its tuples can
have the same SA value. {\em $m$-invariance} \cite{xt07} is a
stricter version of $l$-diversity, by dictating each group to have
exactly $m$ tuples with different SA values. The {\em personalized
approach} \cite{xt06-a} allows each individual to specify her/his
own degree of privacy preservation. The above principles deal with
categorical SAs, whereas $(k, e)$-anonymity \cite{zksy07} and {\em
$t$-closeness} \cite{llv07} support numerical ones. $(k,
e)$-anonymity demands that each QI-group should have size at least
$k$, and the largest and smallest SA values in a group must differ
by at least $e$. $t$-closeness requires that the SA-distribution in
each QI-group should not deviate from that of the whole microdata by
more than $t$.

{\em $\delta$-presence} \cite{nac07} assumes the same background
knowledge as the earlier principles, but ensures a different type of
privacy. It prevents an adversary from knowing whether an individual
has a record in the microdata. {\em $(c, k)$-safety} \cite{mkm+07}
tackles stronger background knowledge. In addition to individuals'
QI values, an adversary may have several pieces of {\em
implicational knowledge}: ``if person $o_1$ has sensitive value
$v_1$, then another person $o_2$ has sensitive value $v_2$". $(c,
k)$-safety guarantees that, if an adversary has at most $k$ pieces of such knowledge, s/he will not be able to infer any individual's SA value with a confidence higher than $c$. Achieving a similar purpose, the {\em
skyline privacy} \cite{crl07} guards against an extra type of
knowledge. Namely, an adversary may have already known the sensitive
values of some individuals before inspecting the published contents.


\extraspacing
\noindent {\bf Generalization algorithms } Numerous heuristic algorithms have
been developed to compute generalization with small information
loss. Although with no provably good worst-case quality or complexity guarantees,
these algorithms are general, since they can be applied to many of
the anonymization principles reviewed earlier, and work
with both numerical and categorical domains. Specifically, a
genetic algorithm is developed in \cite{i02}, and the
branch-and-bound paradigm is employed on a set-enumeration tree in
\cite{ba05,lw10}. Top-down and bottom-up algorithms are presented in
\cite{fwy05,xwp+06}, and the method in \cite{ldr05} borrows ideas
from frequent item set mining. While all the above algorithms adopt single-dimensional generalization, there also exist several multi-dimensional generalization methods. In \cite{ldr06-a}, an algorithm is developed based on a partitioning approach reminiscent of kd-trees. This algorithm is further improved in \cite{ldr06-b} to optimize anonymized data for given workloads. In \cite{gkkm07}, space filling curves are leveraged to facilitate generalization, and the work of \cite{in07} draws an analogy between spatial indexing and generalization. As
shown in \cite{wfwp07}, the previous algorithms may suffer from {\em
minimality attacks}, which can be avoided by introducing some
randomization.

\extraspacing

\noindent {\bf Anonymity in other contexts } The earlier discussion focuses on
data publication, whereas anonymity issues arise in many other
environments. Some examples include anonymized surveying \cite{as00,egs03}, statistical databases \cite{bdm05},
cryptographic computing \cite{jc06}, access control
\cite{bbfs98}, and so on.

\section{Problem Definitions} \label{sec:prob}

Let $\mathcal{T}$ be the raw microdata table, which has $d$
quasi-identifier (QI) attributes $A_1$, ..., $A_d$, and a sensitive
attribute (SA) $B$. Here, $d$ is the {\em dimensionality} of
$\mathcal{T}$, and all attributes are categorical. Given a tuple $t
\in \mathcal{T}$, we employ $t[A_i]$ to denote its $i$-th ($1 \le i
\le d$) QI value, and $t[B]$ its sensitive value. Use $n$ to
represent the cardinality of $\mathcal{T}$, and $m$ to represent the
number of distinct sensitive values in $\mathcal{T}$.  Without loss of
generality we assume that all SA values are from the integer domain $[m] =
\{1, \dots, m\}$.

As in most pervious work on theoretical generalization algorithms, we assume that $\mathcal{T}$ is anonymized with suppression, which can be formally defined based on the concept of {\em partition}. Specifically, a partition $P$ of $\mathcal{T}$ includes disjoint subsets of
$\mathcal{T}$ whose union equals $\mathcal{T}$. We refer to each
subset as a {\em QI-group}. $P$ determines an anonymization
$\mathcal{T}^*$ of $\mathcal{T}$, where all tuples in the same
QI-group carry the same QI values, as shown next.

\begin{definition} [Generalization] \label{def:prob-suppress} \em
A partition $P$ of $\mathcal{T}$ defines a {\em generalization}
$\mathcal{T}^*$ of $\mathcal{T}$ as follows. For each QI-group in
$P$, if all the tuples in the group have the same value on $A_i$ ($i
\in [1, d]$), then they keep this value in $\mathcal{T}^*$;
otherwise, their $A_i$ values are replaced with `*'. All tuples in
$\mathcal{T}$ retain their SA values in $\mathcal{T}^*$.
\end{definition}

For example, Table~\ref{tbl:intro-micro-2-anony} (or
\ref{tbl:intro-micro-2-div}) is a generalization of
Table~\ref{tbl:intro-micro} determined by a partition with 4 (3)
QI-groups. As long as one QI value of a tuple is changed to a star,
we say that this tuple has been {\em suppressed}.

\begin{definition} [$l$-diversity] \label{def:prob-diversity} \em
Given an integer $l$, a set $S$ of tuples is {\em $l$-eligible} if
at most $|S| / l$ of the tuples have an identical SA value. A
generalization $\mathcal{T}^*$ is {\em $l$-diverse} if each QI-group
is $l$-eligible.
\end{definition}

We are ready to define the problem of optimal $l$-diverse
generalization.

\begin{problem} [Star Minimization] \label{prob:prob-star-min} \em
Given a microdata table $\mathcal{T}$ and an integer $l$, find an
optimal $l$-diverse generalization of $\mathcal{T}$ that has the
smallest number of stars.
\end{problem}

Note that there may be multiple optimal solutions with the same
number of stars. An important property of $l$-diversity is {\em
monotonicity}:

\begin{lemma}[\cite{mgk+07}] \label{lmm:prob-mono}
Let $S_1$ and $S_2$ be two disjoint sets of tuples. If both of them
are $l$-eligible, then so is $S_1 \cup S_2$.
\end{lemma}

As an immediate corollary, Problem~\ref{prob:prob-star-min} has a
solution if and only if $\mathcal{T}$ itself is $l$-eligible, i.e.,
at most $|\mathcal{T}| / l$ tuples of $\mathcal{T}$ carry the same
sensitive value. In the following, we focus on only such microdata
tables. It follows that $m \ge l$, where $m$ is the number of distinct sensitive
values in $\mathcal{T}$, as mentioned before.

A close companion of star minimization
(Problem~\ref{prob:prob-star-min}) is {\em tuple minimization}:

\begin{problem} [Tuple Minimization] \label{prob:prob-tuple-min} \em
Given a microdata table $\mathcal{T}$ and an integer $l$, find an
optimal $l$-diverse generalization of $\mathcal{T}^*$ that suppresses the least
number of tuples.
\end{problem}

For instance, in Table~\ref{tbl:intro-micro-2-div}, the amount of
information loss is 8 (stars) in Problem~\ref{prob:prob-star-min},
but 4 (tuples) in Problem~\ref{prob:prob-tuple-min}. Tuple
minimization is different from star minimization because suppressing
various tuples may require different numbers of stars. The following
result builds a connection between the two problems.

\begin{lemma} \label{lmm:prob-reduction}
  A $\lambda$-approximate solution to Problem~\ref{prob:prob-tuple-min} is
  a $\lambda \cdot d$-approximate solution to Problem~\ref{prob:prob-star-min}.
\end{lemma}

\begin{proof}
  Let $\mathcal{T}^*_1$ and $\mathcal{T}^*_2$ be optimal solutions to
  Problems~\ref{prob:prob-star-min} and \ref{prob:prob-tuple-min},
  respectively, and let $\T^*_3$ be a $\lambda$-approximate solution to
  Problem~\ref{prob:prob-tuple-min}.  Use $\alpha_1$ and $\beta_1$ to
  denote the number of stars and the number of tuples suppressed in
  $\mathcal{T}^*_1$, respectively.  Define $\alpha_2, \beta_2$ and
  $\alpha_3, \beta_3$ in the same way for $\mathcal{T}^*_2$ and $\T^*_3$,
  respectively.  Since each suppressed tuple introduces between 1 and $d$
  stars, it holds that $\beta_i \le \alpha_i \le d \cdot \beta_i$ for $i =
  1, 2, 3$.  Hence, $\alpha_3 \le d \cdot \beta_3 \le \lambda \cdot d \cdot \beta_2 \le
  \lambda \cdot d \cdot \beta_1 \le \lambda \cdot d \cdot \alpha_1$.
\end{proof}

In the following sections, we will show that star minimization is
NP-hard when $m \ge 3$, and then, approach this problem through
tuple minimization.

\section{Hardness of Star Minimization} \label{sec:hardness}

As discussed in Section~\ref{sec:intro-cont}, there exists a
straightforward reduction from $l$-diversity to $k$-anonymity. This
reduction, however, works only when the number $m$ of distinct
sensitive values in the microdata table $\mathcal{T}$ equals the
number of tuples in $\mathcal{T}$. It is natural to wonder, in the
more realistic scenario $m \ll |\mathcal{T}|$, is star minimization
(Problem~\ref{prob:prob-star-min}) still NP-hard?


It is easy to observe a polynomial-time algorithm for $m = 2$. In
this case, since $l \le m$, the value of $l$ must be 2 ($l = 1$ is useless for
anonymization). Let $S_1$ ($S_2$) be the set of tuples having the
first (second) SA value. Thus, $|S_1| = |S_2| = |\mathcal{T}| / 2$; otherwise,
$\mathcal{T}$ is not 2-eligible and Problem~\ref{prob:prob-star-min}
has no solution. Then, there exists an $2$-diverse optimal generalization where each QI-group has 2 tuples, since any $2$-diverse QI-group of $\T$ with more than 2 tuples can be divided into smaller $2$-diverse QI-groups, without increasing the number of stars in generalization. Finding this generalization is an instance of bipartite matching.
Specifically, we create a bipartite graph by treating $S_1$ and
$S_2$ as sets of vertices. Draw an edge between each pair $(t_1,
t_2) \in S_1 \times S_2$. The edge has a weight equal to the number
of stars needed to generalize $t_1$ and $t_2$ into the same form. No
edge exists between vertices from the same set. An optimal 2-diverse
generalization corresponds to a minimum perfect matching between
$S_1$ and $S_2$, which can be found in $O(|\mathcal{T}|^3)$ time \cite{k55}.


The above observation has also another implication. In
\cite{wlfw06}, the authors prove that $(\alpha, k)$-anonymity
(explained in Section~\ref{sec:related}) is NP-hard. They do so by
showing that $(0.5, k)$-anonymity is NP-hard. Recall that $(0.5,
k)$-anonymity is essentially the combination of $k$-anonymity and
2-diversity. Intuitively, the hardness of $(0.5, k)$-anonymity stems
from the difficulty of $k$-anonymity. Indeed, the proof in
\cite{wlfw06} no longer holds, when $k$-anonymity is not required.

Next, first assuming $l = 3$, we establish the NP-hardness of star minimization for any $m \ge l$. Later, we will extend the analysis to any $l > 3$. Our derivation is based on a reduction from a classical NP-hard
problem {\em 3-dimensional matching} (3DM) \cite{k72}. Specifically,
let $D_1$, $D_2$, $D_3$ be three dimensions with disjoint domains,
and these domains are equally large: $|D_1| = |D_2| = |D_3| = n$.
The input is a set $S$ of $d \ge n$ distinct 3D points $p_1$, ...,
$p_d$ in the space $D_1 \times D_2 \times D_3$. The goal of 3DM is
to decide the existence of an $S' \subseteq S$ such that $|S'| = n$
and no two points in $S'$ share the same coordinate on any
dimension. For example, assume $D_1 = \{1, 2, 3, 4\}$, $D_2 = \{a,
b, c, d\}$, and $D_3 = \{\alpha, \beta, \gamma, \delta\}$, and a set
$S$ of 6 points in Figure~\ref{fig:hardness-ex}a. Then, the result
of 3DM is ``yes": a solution $S'$ can be $\{p_1, p_3, p_5, p_6\}$.


\begin{figure}[t]
\centering
\includegraphics[height=7mm]{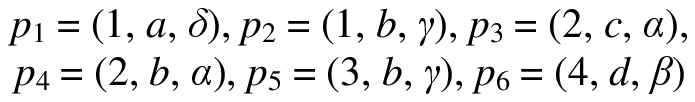} \\[1mm]
(a) The contents of $S$ \\[3mm]
\includegraphics[height=40mm]{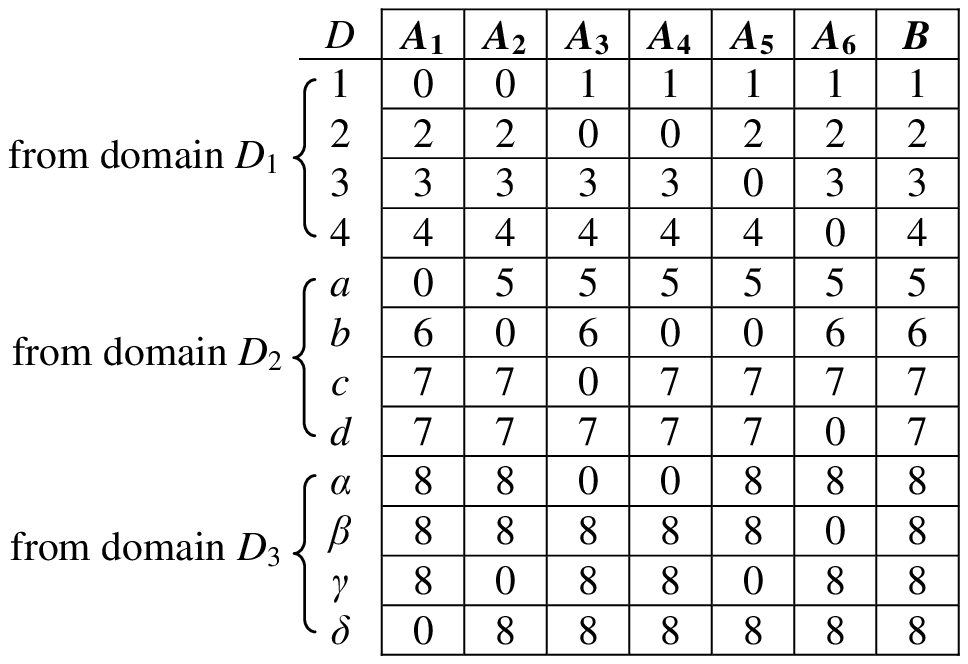} \\[1mm]
(b) The constructed table $\mathcal{T}$ ($m=8$)
\figcapup
\caption{Illustration of reduction}
\figcapdown
\label{fig:hardness-ex}
\end{figure}

%
%

Let $v_1, ..., v_n$ be the values in $D_1$, $v_{n+1}, ..., v_{2n}$ be the values in $D_2$, and $v_{2n+1}, ..., v_{3n}$ be the values in $D_3$. We construct a microdata table $\mathcal{T}$ from $S$. Specifically, $\mathcal{T}$ has
\begin{itemize}
\item a sensitive attribute $B$;

\item $d$ QI attributes $A_1$, $A_2$, ..., $A_d$, where $A_i$
($1 \le i \le d$) corresponds to the $i$-th point $p_i$ in $S$;

\item $3n$ rows, where the $j$-th ($1 \le j \le 3n$) row corresponds to $v_j$.
\end{itemize}
The rows in $\mathcal{T}$ are constructed as follows. Let $t$ be the $j$-th ($1 \le j \le 3n$) row of $\mathcal{T}$. We first select a positive integer $u$ according to the value of $j$ (details to be clarified shortly). Then, we set the SA value of $t$ to $u$, i.e., $t[B] = u$. After that, for each $i \in [1, d]$, we set $t[A_i]$ to $0$ if $v_j$ is a coordinate of point $p_i \in S$, or $u$ otherwise. Because each $p_i$ has three coordinates, the following property of $\mathcal{T}$ holds.
\begin{property} \label{lmm:hardness-three-0}
For any $i \in [1, d]$, there exist exactly $3$ rows in $\mathcal{T}$ that have value $0$ on $A_i$.
\end{property}

The value of $u$ is chosen in a way that ensures another two properties of $\mathcal{T}$. First, $\mathcal{T}$ should contain $m$ ($m \le 3n$) distinct SA values, namely, $1$, $2$, ..., $m$. Second, for any $i, j \in [1, 3n]$, if $v_i$ and $v_j$ belong to different domains (e.g., $v_i \in D_1$ and $v_j \in D_2$), the $i$-th and $j$-th rows in $\mathcal{T}$ should have different SA values.

Specifically, we set $u = j$ for any $j \in [1, m\!-\!2]$. When $j \in [m\!-\!1, 3n]$, we differentiate three cases according to the values of $m$ and $n$:
\begin{itemize}
\item If $m\!-\!1 > 2n$, we let $u = m\!-\!1$ if $j \in [m\!-\!1, 3n\!-\!1]$, and $u = m$ if $j = 3n$.

\item If $2n \ge m\!-\!1 > n$, then $u = m\!-\!1$ if $j \in [m\!-\!1, 2n]$, and $u = m$ if $j \in [2n\!+\!1, 3n]$.

\item If $n \ge m\!-\!1$, we set (i) $u = m\!-\!2$ if $j \in [m\!-\!1, n]$, (ii) $u = m\!-\!1$ if $j \in [n\!+\!1, 2n]$, and (iii) $u = m$ if $j \in [2n\!+\!1, 3n]$.
\end{itemize}

Figure~\ref{fig:hardness-ex}b demonstrates the $\mathcal{T}$ built from the $S$ in Figure~\ref{fig:hardness-ex}a, when $m = 8$. For example, let $t$ be the $7$-th row (i.e., $j = 7$), which corresponds to value $c \in D_2$. Since $j=7$, $n=4$, and $m=8$, we have $2n \ge m\!-\!1 > n$ and $j \in [m - 1, 2n]$. Hence, $u=m\!-\!1 = 7$. $t[A_3]$ equals 0, because $c$ is the second coordinate of $p_3 \in S$. $t$ has 7 on other QI attributes because $c$ is not the 2nd coordinate of any other point in $S$.

Let $\mathcal{T}^*$ be any 3-diverse generalization of
$\mathcal{T}$. We say that a QI-group $Q$ in $\mathcal{T}^*$ is {\em
futile} if all the QI values in $Q$ are stars (i.e., $Q$ retains no
QI information at all). Otherwise, $Q$ is {\em useful}. $\mathcal{T}^*$ have several properties.


%

\begin{property} \label{lmm:hardness-zero}
If a QI-group $Q$ in $\mathcal{T}^*$ is useful, then all non-star QI values in $Q$ must be $0$.
\end{property}
\begin{proof}
Consider any $i \in [1, d]$, such that $Q$ has no star on $A_i$ after generalization. Then, before generalization, all tuples in $Q$ should have the same value on $A_i$. Let this value be $x$. By the way $\mathcal{T}$ is constructed, if $x \ne 0$, all tuples in $Q$ should have an SA value $x$, which contradicts the assumption that $Q$ is $3$-eligible. Therefore, $x = 0$ holds.
\end{proof}

\begin{property} \label{lmm:hardness-exactly-3}
Any useful QI-group $Q$ in $\mathcal{T}^*$ contains (i) exactly three tuples, (ii) $3(d-1)$ stars, and (iii) $3$ zeros.
\end{property}
\begin{proof}
Let $h$ be the number of tuples in $Q$. Since $Q$ is useful, there exists $i \in [1, d]$, such that all tuples in $Q$ have value $0$ on $A_i$ (see Property~\ref{lmm:hardness-zero}). By Property~\ref{lmm:hardness-three-0}, there exist only three tuples in $\mathcal{T}$ that have value $0$ on $A_i$. Hence, $h \le 3$. On the other hand, because $Q$ is $3$-eligible, $h \ge 3$. Therefore, $Q$ contains exactly three tuples.

Consider that, before generalization, $Q$ contains three rows $t_a$, $t_b$, and $t_c$ ($a, b, c \in [1, 3n]$) in $\mathcal{T}$. Assume on the contrary that there exists $y \in [1, d]$, $y \ne i$, such that $Q$ has no star on $A_y$. Recall that, the $j$-th ($j \in [1, 3n]$) row in $\mathcal{T}$ has value $0$ on $A_i$ ($A_y$), if and only if $v_j$ is a coordinate of $p_i$ ($p_y$). Thus, each of $v_a$, $v_b$, and $v_c$ should appear in both $p_i$ and $p_y$. This indicates that $p_i = p_y$, leading to a contradiction. Therefore, $Q$ should have $0$ on exactly one QI attribute. Since $Q$ contains three tuples, the number of stars (zeros) in $Q$ should be $3d - 3$ (3).
\end{proof}


\begin{property} \label{lmm:hardness-useful-tuple}
$\mathcal{T}^*$ has at least $3n (d - 1)$ stars.
\end{property}
\begin{proof}
Let us analyze the number of non-star QI values in $\mathcal{T}^*$.
Each non-star can come from only a useful QI-group. According to
Property~\ref{lmm:hardness-exactly-3}, each such group contains 3 non-star QI values. As $\mathcal{T}^*$ has $3n$ rows and each useful QI-group
has 3 rows, there can be at most $n$ useful QI-groups. Therefore,
the number of non-star QI values is at most $3n$, and the property
follows.
\end{proof}

\begin{lemma} \label{lmm:hardness-if-and-only-if}
The 3DM on $S$ returns ``yes", if and only if there is a 3-diverse
generalization of $\mathcal{T}$ with $3n (d - 1)$ stars.
\end{lemma}
\begin{proof}
``{\em Only-if direction}": Without loss of generality, let $S' =
\{p_1, ..., p_n\}$ be the solution of the 3DM. Then, we create $n$
useful QI-groups $Q_1$, ..., $Q_n$, where $Q_i$ ($1 \le i \le n$)
encloses the 3 tuples in $\mathcal{T}$ whose values on attribute
$A_i$ are 0. By the way $\mathcal{T}$ is constructed, each of the 3 tuples corresponds to a
coordinate of $p_i$, and has a distinct SA value. Hence, $Q_i$ is
3-eligible. Since the points in $S'$ do not share any coordinate,
$Q_1$, ..., $Q_n$ are mutually disjoint, and hence, their union
covers the entire $\mathcal{T}$ (which has $3n$ tuples in total).
Generalizing each $Q_i$ ($1 \le i \le n$) introduces $3(d-1)$ stars,
leading to totally $3n(d-1)$ stars. The resulting generalization is
3-diverse.

``{\em If direction}": Let $\mathcal{T}^*$ be a 3-diverse
generalization with $3n(d-1)$ stars. According to
Property~\ref{lmm:hardness-exactly-3}, $\mathcal{T}^*$ has exactly $n$
QI-groups, and all of them are useful. Denote these groups as $Q_1$,
..., $Q_n$. Since each group contributes $3(d-1)$ stars, $Q_x$ ($1
\le x \le n$) has no star on exactly one QI-attribute $A_i$ ($i \in [1, d]$). Let us call
$A_i$ the {\em useful QI-attribute} of $Q_x$. Assume that, before generalization, $Q_x$ contains three rows $t_a$, $t_b$, and $t_c$ ($a, b, c \in [1, 3n]$), where $t_j$ denotes the $j$-th row in $\mathcal{T}$. Since $t_a[A_i] = t_b[A_i] = t_c[A_i] = 0$, according to the way $\mathcal{T}$ is generated, $v_a$, $v_b$, and $v_c$ should be the coordinates of $p_i$. We define $p_i$ as the point in $S$ {\em corresponding to} $Q_x$. Observe that, the useful QI attributes of any two QI-groups must be different, otherwise there exist at least six tuples in $\mathcal{T}$ that have $0$ on the same QI attribute, contradicting Property~\ref{lmm:hardness-three-0}. Consequently, each QI-group should correspond to a distinct point in $S$. Without loss of generality, assume that $p_j$ is the point corresponding to $Q_j$ ($1 \le j \le n$). Let $S' = \{p_1, ..., p_n\}$. Since each row in $\mathcal{T}$ appears in exactly one QI-group, each coordinate in $D_1 \cup D_2 \cup D_3$ appears in exactly one point in $S'$. Therefore, $S'$ is a solution to the 3DM.
\end{proof}

Property~\ref{lmm:hardness-useful-tuple} and
Lemma~\ref{lmm:hardness-if-and-only-if} imply that we can decide
whether $S$ has a 3DM solution, by examining if an optimal 3-diverse
generalization of $\mathcal{T}$ has $3n (d - 1)$ stars. Therefore,
if we had an optimal polynomial-time $l$-diversity algorithm that
works on all microdata tables with $m \in [l, |\mathcal{T}|]$, this algorithm would
also solve 3DM in polynomial time.

Extending the above analysis in a straightforward manner, we can
show that, for any $l > 3$, optimal $l$-diversity under the
constraint $m \ge l$ is also NP-hard, through a reduction from {\em
$l$-dimensional matching} \cite{hss03}. Thus, we arrive at:

\begin{theorem} \label{thrm:hardness-NP}
For any $m \ge l \ge 3$, optimal $l$-diverse generalization
(Problem~\ref{prob:prob-star-min}) is NP-hard.
\end{theorem}

We conclude this section by pointing out that our proof requires
only an alphabet (i.e., the union of the domains of all attributes
in the microdata) of size $m + 1$. For example, in
Figure~\ref{fig:hardness-ex}, $m = 8$ and $\mathcal{T}$ has 9
different values 0, 1, ..., 8.



\section{Tuple Minimization} \label{sec:tmin}

This section tackles tuple minimization
(Problem~\ref{prob:prob-tuple-min}), and presents an
algorithm with an approximation ratio of $l$. By
Lemma~\ref{lmm:prob-reduction}, it leads to an $(l \cdot
d)$-approximation for star minimization, resulting in the first
$l$-diversity algorithm with a non-trivial worst-case bound on
information loss. Furthermore, this algorithm leverages several
novel heuristics that work fairly well in practice, and usually
produce a solution with a much better quality than the upper bound.

\subsection{Algorithm Overview} \label{sec:algorithm}

Since our goal is to minimize the number of tuples suppressed, we
can redefine the problem as the following.  Suppose that the
microdata $\mathcal{T}$ is partitioned into $s$ QI-groups $Q_1,
\dots, Q_s$, where tuples in the same QI-group have the same value
on every QI attribute. The problem is to remove the minimum number
of tuples from $Q_1, \dots, Q_s$, such that: (a) all QI-groups are
$l$-eligible, and (b) the set of all removed tuples is $l$-eligible.
We denote the set of removed tuples by $R$, and these tuples will
correspond to the suppressed tuples. We refer to $R$ as the {\em residue set}. Since switching tuples with the
same QI and SA values will not change the quality of the solution, in the
following we will not distinguish such tuples.  In this manner, the
QI-groups and $R$ are effectively considered as multisets.

Initially $R$ is empty.  Throughout the algorithm, tuples are only
moved to $R$ but never taken back.  We follow different rules to
pick tuples to remove in the three phases of the algorithm.  In the
first phase, we will make sure that condition (a) above is
satisfied.  If condition (b) is also met, the algorithm immediately
terminates.  Otherwise in phase two, we try to do an ``easy fix'' of
the problem by removing some more tuples from the QI-groups without
violating condition (a).  If at any time during phase two condition
(b) is met, the algorithm ends, or else we proceed to phase three.
In the last phase, we do an ``overhaul'' in order to satisfy
condition (b) by removing tuples in large batches.  Before giving
the details for each phase below, we point out that approximations
are introduced in succession: If the algorithm terminates during the
first phase, then the returned solution is guaranteed to be optimal;
if the algorithm ends in phase two, only an additive error of $l-1$
is introduced; only in phase three may we encounter a multiplicative
error of $l$.

Sections~\ref{sec:algorithm-1}-\ref{sec:algorithm-3} describe the
conceptual procedures of the three phases respectively, and analyze
their theoretical guarantees. We defer the running time discussion
to Section~\ref{sec:implementation}.

\subsection{Phase One} \label{sec:algorithm-1}

For a QI-group $Q$ and an SA value $v$, denote the number of tuples
in $Q$ with SA value $v$ by $h(Q, v)$.  We call the SA value with
the most tuples the {\em pillar SA value}, or simply the {\em
pillar}.  The number of tuples in the pillar is called the {\em
pillar height} of $Q$, denoted by $h(Q) = \max_v h(Q,v)$.  Note that
there could be more than one pillar in a QI-group.  These terms are
similarly defined on the set of removed tuples $R$.

\extraspacing \noindent {\bf Algorithm } The rule of phase one is simple: for each QI-group,
repeatedly remove one tuple from its pillar until the QI-group is
$l$-eligible. If there is more than one pillar, the choice can be
arbitrary.  Note that although we break ties arbitrarily, the end
result is unique, the reason being the following. When there are
more than one pillar in the QI-group that is still not $l$-eligible,
removing a tuple from any of the pillars will not decrease the
pillar height, and hence the QI-group will not become $l$-eligible
after the removal. Only after all pillars have lost one tuple does
the QI-group have a chance of becoming $l$-eligible.  In other
words, no matter what order is taken, we will eventually remove one
tuple from each pillar.

At the end of phase one, we check if $R$ is $l$-eligible, i.e.,
\begin{equation}
\label{eq:1} |R| \ge l \cdot h(R).
\end{equation}
If (\ref{eq:1}) holds, then the algorithm terminates; otherwise we
proceed to phase two.

Consider the example in Table~\ref{tbl:intro-micro} with $l=2$.
Initially we have 4 QI-groups: $\{1,2\}, \{3\}, \{4\}, \{5,6,7,8\},
\{9,10\}$.  After phase one, the first three QI-groups are completely
eliminated, and the other two QI-groups remain unchanged.  The set
$R$ of removed tuples have the following (multi)set of SA values:
$\{$HIV, HIV, pneumonia, bronchitis$\}$. In this case $R$ is already
$l$-eligible and thus the whole algorithm terminates.  However, if
we are not so lucky, we need to go to phase two.

\extraspacing \noindent {\bf Analysis } Let $\Qo_1, \dots, \Qo_s$ be the QI-groups at the
end of phase one, and $\Ro$ the set of removed tuples.  If
(\ref{eq:1}) holds after phase one, then $(\Qo_1, \dots, \Qo_s,
\Ro)$ must be an optimal solution. In fact, we can prove a stronger
result.

\begin{lemma}
\label{lem:phase-one}
  For any $1\le i \le s$ and any subset $Q'_i$ of $Q_i$ that is
  $l$-eligible, $h(Q'_i, v) \le h(\Qo_i, v)$ for all SA values $v$.
\end{lemma}

The proofs for the rest of the paper can be found in the appendix. Based on Lemma~\ref{lem:phase-one}, we prove the following corollary.

\begin{corollary}
\label{cor:phaseone}
  If the algorithm terminates after phase one, then
  $(\Qo_1,\dots,\Qo_s,\Ro)$ is an optimal solution.
\end{corollary}

Let $OPT$ be the number of tuples in $R$ in the optimal solution.
The following lower bound on $OPT$ is another easy corollary of
Lemma~\ref{lem:phase-one} that will be useful later on.

\begin{corollary}
\label{cor:optlower}
  $OPT \ge l \cdot h(\Ro)$.
\end{corollary}

\subsection{Phase Two} \label{sec:algorithm-2}
In phase two, we try to increase $|R|$ while keeping $h(R)$
unchanged by removing tuples from the QI-groups while maintaining
their $l$-eligibility. We continue the process until
Inequality~(\ref{eq:1}) is satisfied, or no more tuples can be
removed.

Before describing the phase two algorithm we need some more
terminology. We know that $|Q| \ge l \cdot h(Q)$ for any QI-group
$Q$ at the end of phase one.  We say that $Q$ is {\em thin} if $|Q|
= l \cdot h(Q)$, and {\em
  fat} if $|Q| \ge l \cdot h(Q) + 1$.  If $Q$ has one or more pillars that
are also the pillars of $R$, then $Q$ is a {\em conflicting}
QI-group; these pillars are the {\em conflicting pillars} of $Q$. If
$Q$ is both thin and conflicting, it is said to be {\em dead};
otherwise it is {\em alive}.  Intuitively, a dead QI-group cannot
lose any more tuples without either increasing $h(R)$ or violating
its own $l$-eligibility.  An SA value $v$ is {\em alive} if there
exists at least one alive QI-group $Q$ such that $h(Q,v) > 0$.

\extraspacing \noindent {\bf Algorithm } Phase two proceeds iteratively as follows. In each
iteration, we pick an alive SA value $v$ such that $h(R,v)$ is
minimized, i.e., $v$ is the least frequent alive SA
value in $R$. When there are multiple such SA values, we pick one
arbitrarily.  If there is no alive SA value, then phase two cannot
solve the problem and we enter phase three.  Otherwise, we go to the
QI-group $Q$ where $h(Q,v) > 0$; again the choice is arbitrary if
there is more than one option.  There are two cases: If $Q$ is fat,
then we simply remove a tuple from $Q$ with SA value $v$,
decrementing $h(Q,v)$ while incrementing $h(R,v)$.  If $Q$ is thin,
then by definition it must be non-conflicting, so we remove a tuple
from each of $Q$'s pillars. Note that this may or may not increase $h(R, v)$. This iteration now ends. If at this time
$R$ becomes $l$-eligible, the whole algorithm terminates; otherwise
a new iteration starts.

Consider the following example with $m = 5$ SA values, $s=3$
QI-groups, $l = 3$, and $Q_1 = (3,1,1,2,3), Q_2 = (0, 2, 2, 4, 4),
Q_3 = (4, 4, 0, 0, 0)$ before phase one.  (For notational simplicity
we use the vector presentation for multisets.  For instance,
(3,1,1,2,3) means there are three tuples with SA value 1, one tuple
with SA value 2 and 3 respectively, two and three tuples with SA
values 4 and 5 respectively.) In phase one, $Q_1$ and $Q_2$ do not
change, while all tuples of $Q_3$ are removed. Thus at the end of
phase one the status is $Q_1 = (3,1,1,2,3), Q_2 = (0, 2, 2, 4, 4), R
= (4, 4, 0, 0, 0)$. Now phase two starts.
In the first iteration, there are five alive SA values: 1, 2, 3, 4, and 5.
Suppose we pick
$v=3$. Note that both $Q_1$ and $Q_2$ can give a 3 to $R$, and the
choice can be arbitrary.  Say we remove a 3 from $Q_1$, changing the
status to $Q_1 = (3,1,0,2,3), Q_2 = (0, 2, 2, 4, 4), R = (4, 4, 1,
0, 0)$. Now $Q_1$ is dead, since it is both thin and conflicting
(the conflicting pillar is 1). In the second iteration, still 3, 4,
5 are all alive SA values, but since 4 and 5 have the minimum
$h(R,v)$, we pick one of them arbitrarily, say 4.  $Q_2$ now is the
only alive QI-group: it is thin but non-conflicting.  So we remove a
4 and a 5 together, changing the status to $Q_1 = (3,1,0,2,3), Q_2 =
(0, 2, 2, 3, 3), R = (4, 4, 1, 1, 1)$.  In the third iteration, 3,
4, 5 are all possible choices, and $Q_2$ is fat.  Say we remove a 3,
which results in $Q_1 = (3,1,0,2,3), Q_2 = (0, 2, 1, 3, 3), R = (4,
4, 2, 1, 1)$.  At this point $R$ has become $l$-eligible, and the
algorithm terminates.

\extraspacing \noindent {\bf Analysis } Let $\Qt_1, \dots, \Qt_s, \Rt$ be the status at the
end of phase two.  We first prove that $h(R$) does not increase in
this phase, that is:

\begin{lemma} \label{lmm:rt=ro}
$h(\Rt) = h(\Ro)$.
\end{lemma}

In general, $\Rt$ may not be $l$-eligible.  However, if it is, then
the following guarantee holds.

\begin{lemma}
\label{lem:phasetwo} If the algorithm terminates during phase two,
then $|\Rt| \le l \cdot h(\Ro) + l -1$.
\end{lemma}

By combining Corollary~\ref{cor:optlower} and Lemma~\ref{lem:phasetwo}, we have the following result.

\begin{corollary} \label{cor:twophase-lower}
If the algorithm terminates during phase two, then it returns a
solution such that $|\Rt| \le OPT + l -1$.
\end{corollary}

It is clear that all QI-groups are dead after phase two (unless the
algorithm terminates).  In this case, the following property holds,
which will be useful later on.

\begin{lemma}
  \label{lem:Rt}
  If $Q_1, \dots, Q_s$ are all dead and $R$ is not $l$-eligible, then for
  any pillar $p$ of $R$, there exists some $Q_i$ such that $p$ is not a
  conflicting pillar of $Q_i$.
\end{lemma}

The following corollary follows from Lemma~\ref{lem:Rt}.

\begin{corollary}
\label{cor:Rt} If the algorithm does not terminate after phase two,
then $\Rt$ has at least two pillars.
\end{corollary}

The result above implies that for the special case $l=2$, the
algorithm always terminates during the first two phases.

\begin{theorem}
\label{thm:l2} For $l=2$, our algorithm always solves the tuple
minimization problem during the first two phases with a solution
$|\Rt| \le OPT +1$.
\end{theorem}

\subsection{Phase Three} \label{sec:algorithm-3}
In most cases the algorithm will stop in the first two phases.
However, on some ``hard'' inputs we will have to resort to phase
three.  In the output from phase two, all QI-groups are thin and
conflicting, and still $|R| < l \cdot h(R)$.  The failure of phase
two suggests that in order to satisfy Inequality~(\ref{eq:1}), we
cannot just increase $|R|$.  We need to increase both $|R|$ and
$h(R)$, but in a careful way such that the amount of increase in
$|R|$ is more than $l$ times the increase in $h(R)$, so that
eventually the gap between $|R|$ and $l \cdot h(R)$ can be closed.

\extraspacing \noindent {\bf Algorithm } The third phase proceeds in rounds, each consisting
of two steps.  In the first step, we pick a subset $S$ of QI-groups,
and remove one tuple from each of their pillars.  This increases
$h(R)$ but also (possibly) makes these QI-groups fat. Meanwhile,
since certain pillars of $R$ might have disappeared after this step,
some other QI-groups may switch from conflicting to non-conflicting.
More precisely, a greedy algorithm is used to decide $S$. Initially
we set $P$ to be the set of pillars of $R$. For a QI-group $Q$, let
$C(Q)$ be the set of conflicting pillars of $Q$. The greedy
algorithm iteratively does the following. As long as $P$ is not
empty, pick the QI-group $Q$ that minimizes $|C(Q) \cap P|$, and
then set $P \leftarrow P \cap C(Q)$. Note that here the problem is
equivalent to {\sc
  set cover} \cite{clr01}, i.e., we are using $\overline{C(Q)}$ as the ``sets''
to cover all the pillars of $R$, and this greedy algorithm is
actually the same as the standard heuristic for {\sc set cover}.

In the second step, for each QI-group $Q$, if it has become alive
after step one, then we keep removing tuples from $Q$ until it
becomes dead again, using the following simple rules: If $Q$ is fat,
remove a tuple from any SA value that is not a pillar of $R$. If $Q$
is thin, then check if it is conflicting.  If yes, then we are done
with this QI-group; otherwise we remove a tuple from each of its
pillars.  If at any time $R$ becomes $l$-eligible, the whole
algorithm terminates.  Note that if the algorithm does not terminate
after a round, all QI-groups have become dead again.

The following is an example showing how phase three works.  Suppose
$m=5, s=2, l=4$, and the status after phase two is $Q_1 =
(3,1,2,3,3), Q_2 = (1,3,2,3,3), R=(4,4,4,0,0)$. Note that $Q_1$ and
$Q_2$ are both thin and conflicting: $Q_1$ conflicts on 1 while
$Q_2$ conflicts on 2. In step one of the first round, we pick
QI-groups whose $\overline{C(Q)}$ together cover the pillars
$\{1,2,3\}$ of $R$. As $\overline{C(Q_1)} = \{2, 3, 4, 5\}$ and
$\overline{C(Q_2)} = \{1, 3, 4, 5\}$, the greedy algorithm chooses
both $Q_1$ and $Q_2$.  Then we remove one tuple from each of the
pillars of $Q_1$ and $Q_2$, resulting in the following
configuration: $Q_1 = (2,1,2,2,2), Q_2 = (1,2,2,2,2),
R=(5,5,4,2,2)$.  In step two, we first remove tuples from $Q_1$
until it becomes dead. As $Q_1$ is fat and SA values
3, 4, 5 are not a pillar of $R$, we can remove any tuple of those SA
values from $Q_1$. Suppose a 3 is removed, resulting in $Q_1 =
(2,1,1,2,2), Q_2 = (1,2,2,2,2), R=(5,5,5,2,2)$. Similarly we remove
a 4 from $Q_2$, leading to $Q_1 = (2,1,1,2,2), Q_2 = (1,2,2,1,2),
R=(5,5,5,3,2)$. At this point, $R$ is $l$-eligible and the algorithm
terminates.  In this simple example, there is only one round, but in
general there could be multiple rounds.

\extraspacing \noindent {\bf Analysis } We now analyze the approximation ratio guaranteed by
the algorithm. As it turns out, the key factor is to bound the increase in $h(R)$
throughout phase three.

\begin{lemma}
\label{lem:round1}
  In each round of phase three, $h(R)$ increases by at most $l-2$.
\end{lemma}

\begin{lemma}
  \label{lem:round2}
  There are at most $h(\Rt)$ rounds in phase three.
\end{lemma}

Based on Lemmas \ref{lem:round1} and \ref{lem:round2}, we prove our main theorem as follows.

\begin{theorem}
  \label{thm:main}
  Our algorithm finds an $l$-approximate solution to the tuple minimization
  problem.
\end{theorem}

\subsection{Implementation} \label{sec:implementation}

Our three-phase algorithm can be implemented efficiently using
inverted list structures.  In this subsection, we present an
implementation, which has a worst-case time complexity of $O(s \cdot
n)$.

\extraspacing \noindent {\bf The basic data structure } We maintain an array $\A_i$ for each
QI-group $Q_i$ throughout the algorithm, as well as an $\A_R$ for
the set of removed tuples $R$. Suppose that $Q_i$ has $n_i$ tuples,
for $i=1,\dots, s$.  The array $\A_i$ has $n_i$ entries.  The $j$-th
entry, $\A_i[j]$, contains a pointer to a list of SA values $v$ such
that $h(Q_i, v) = j$.  Note that some entries of $A_i$ may be empty.
Along with each SA value $v$, we keep a pointer to a list of tuples
in $Q_i$ with this SA value, called the {\em SA set} of $v$. For
each $\A_i$, we also maintain $p_i$, the maximum index $j$ such that
$\A_i[j]$ is nonempty.  In other words, $\A_i[p_i]$ always points to
the list of pillars of $Q_i$.  We similarly maintain the pillar
pointer $p_R$ for $\A_R$.  The whole data structure uses linear
space and can also be easily initialized in $O(n)$ time.

This data structure supports an update, i.e., moving a tuple from
some $Q_i$ to $R$ in constant time.  To move a tuple $t$, we first
remove it from its SA set, stored at some $\A_i[j]$.  If $j=1$, we
also delete the SA set; otherwise, we move the SA set from $\A_i[j]$
to $\A_i[j-1]$.  Next, we insert $t$ to $\A_R$, and the procedure is
symmetric.  Finally, we also update $p_i$ and $p_R$.  Note that
although $p_i$ may decrease a lot in a single update, the amortized
cost of maintaining $p_i$ is $O(1)$, since $p_i$ only moves in
one direction, and the total distance it travels is at most $n_i$.

\extraspacing \noindent {\bf Phase one } Consider the QI-group $Q_i$ with $n_i$ tuples.  In
phase one, we simply keep removing tuples from the pillar of $Q_i$,
i.e., the first SA set in the list pointed by $\A_i[p_i]$.  Since
the update cost for each removed tuple is $O(1)$, and we can also
easily check if $Q_i$ is $l$-eligible after each update, the running
time for this QI-group is $O(n_i)$, implying a total running time of
$O(n)$ for phase one.


\extraspacing \noindent {\bf Phase two } To efficiently implement our phase two algorithm,
another inverted list $\C$, called the {\em candidate list}, is
required.  It is an array of size $n$.  At $\C[j]$ we store a list
of entries of the form $(i, v)$, one for each alive SA value $v$ in
$Q_i$ if $h(R,v) = j$.  That is, the list at $\C[j]$ stores (the
pointers to) all possible SA sets from which we can remove tuples.
$\C$ can be initialized in $O(n)$ time.  It can also be maintained
with $O(1)$ cost after a tuple is inserted to $R$.

In each iteration of phase two, we pick a pair $(i,v)$ from the list
stored at the first non-empty entry $\C[j]$.  Next we check if $Q_i$
is fat.  If it is we simply remove a tuple from $Q_i$ with SA value
$v$; otherwise we remove a tuple from each of $Q_i$'s pillars.  At
the end of the iteration, we check if $Q_i$ is dead.  If so we
remove all its entries $(i,v)$ from $\C$.  Since the cost to remove
a tuple is $O(1)$, and there are at most $n$ entries in $\C$, the
total cost of phase two is $O(n)$.


\extraspacing \noindent {\bf Phase three } The first step of each round in phase three is
the standard greedy algorithm for {\sc set cover}, which can be
implemented in $O(s \cdot l)$ time \cite{clr01}, since there are $s$
sets and each set has cardinality at most $l$.  In the second step,
by using the inverted list $\A_i$, a QI-group $Q_i$ can be handled
in time $O(l + r)$, where $r$ is the number of tuples removed from
$Q_i$.  To see this, note that every time we apply the rule, we
either remove tuples, whose cost can be charged to the $O(r)$ term,
or declare that $Q_i$ is dead, whose cost is at most $O(l)$.  Since
we remove at most $n$ tuples in total, the overall cost of phase
three is thus $O(sl \cdot h(\Ro) + n)$ as there are at most $h(\Ro)$
rounds by Lemma~\ref{lem:round2}.  Finally, since $h(\Ro) \le n/l$,
we conclude that the total cost of phase three is $O(s \cdot l
\cdot n/l + n) = O(s \cdot n)$.  Note that this is a very
pessimistic bound, as the typically number of rounds is much smaller
than $n/l$ in practice.


\begin{theorem}
Our three-phase algorithm can be implemented in $O(s \cdot n)$ time.
\end{theorem}

\subsection{Discussions} \label{sec:algorithm-discuss}

The performance of our algorithm is sensitive to the diversity of QI values in the microdata. If most tuples in the microdata have distinct QI values, the first phase of our algorithm would start with a large number of QI groups that contain less than $l$ tuples; eventually, all tuples in these QI groups will be moved to the set $R$ and suppressed, leading to a significant number of stars in the generalized data. Such degradation of data utility usually occurs when the microdata contains QI attributes with large domains. For example, a micordata table with {\em Birth Date}, {\em Gender}, and {\em ZIP Code} as the QI attributes would contain a significant number of tuples that have distinct QI values, since both {\em Birth Date} and {\em ZIP Code} have sizable domains, and hence, any two tuples are likely to differ on either attribute\footnote{Indeed, a recent study \cite{s02-b} has shown that $87\%$ of the U.S. population can be uniquely identified by their birth dates, genders, and 5-digit ZIP codes.}.

Despite the above drawback, our algorithm can still be useful in some scenarios, due to the following reasons. First, our algorithm can be applied on datasets with small or median QI domains. Such microdata exists, as many QI attributes in practice, such as {\em Gender}, {\em Race}, {\em Marital Status}, {\em Years of School Attendance}, have domains with cardinalities below $20$.

Second, QI attributes with large domains often need to be coarsened (even before generalization is performed) to avoid disclosure of excessively detailed personal information. For example, the {\em Standards for Privacy of Individually Identifiable Health Information} \cite{hipaa} (issued by the U.S. Department of Health and Human Services) requires that, unless otherwise justified, any personal data to be published should satisfy the following two conditions (in addition to numerous other requirements):
\begin{enumerate}
\item given any date directly related to an individual (e.g., birth date, admission date, discharge date), only the year of the date is released;

\item only the first three digits of any ZIP code are retained.
\end{enumerate}
Therefore, given a dataset with QI attributes {\em Birth Date} and {\em ZIP Code}, if the publisher is to release the data in a manner that conforms to the above standard, s/he should transform {\em Birth Date} to {\em Year of Birth}, and remove all but the initial three digits of any ZIP code. This considerably reduces the domain size of the attributes, making our algorithm applicable on the dataset.



Third, our algorithm can be easily combined with any heuristic suppression algorithm to improve its performance over datasets with diverse QI values. Specifically, given a micordata table, we can first employ our algorithm to obtain (i) a set of QI-groups that contain no stars, and (ii) the residue set $R$. After that, we can apply any existing heuristic algorithm on $R$ to divide it into smaller QI-groups, thus reducing the number of values that need to be suppressed. Apparently, such a hybrid approach always outperforms our algorithm in star minimization, and hence, it also achieves an approximation ratio of $O(l \cdot d)$.

Last but not least, given a microdata table, we may preprocess it with any single-dimensional generalization method to reduce the cardinalities of the QI domains, and then apply our algorithm on the modified dataset. The preprocessing step method does not need to ensure $l$-diversity: even the $k$-anonymity algorithms \cite{i02,ba05,fwy05,wyc04,ldr05} can be applied. The amount of generalization imposed in the preprocessing step has an effect on the quality of the $l$-diverse table output by our algorithm. In particular, less generalization leads to large domains of the QI attributes, which, in turn, results in more stars in the $l$-diverse table. On the other hand, when the QI attributes are coarsened to a higher degree during preprocessing, each non-star QI value in the $l$-diverse tale corresponds to a larger sub-domain of the QI attribute, i.e., the published QI values are less accurate. To achieve a good tradeoff between the number of stars and the accuracy of non-star QI values, we may vary the amount of generalization in the preprocessing step, examine the output of our algorithm, and choose the setting that optimizes the utility of the $l$-diverse table. A complete treatment of this issue, however, is beyond the scope of this paper.

\section{Experiments} \label{sec:exp}

This section experimentally evaluates the proposed techniques. Section~\ref{sec:exp-star} examines the performance of our algorithms in star minimization, and Section~\ref{sec:exp-kl} compares our algorithms with single-dimensional generalization methods. All of our experiments are performed on a computer with a 3 GHz Pentium IV CPU and 2 GB RAM.

\subsection{Star Minimization} \label{sec:exp-star}

\noindent {\bf Algorithms evaluated } The existing $l$-diversity techniques employ either single- or multi-dimensional generalization. We examine the state of the art \cite{fwy05,ldr06-a,gkkm07} of these techniques, modify them as suppression algorithms, and choose {\em Hilbert} \cite{gkkm07}, the one that achieves the best performance in star minimization, as the baseline with which our algorithms are compared. We denote the three phase algorithm in Section~\ref{sec:algorithm} as {\em TP}. We have also implemented a hybrid algorithm, {\em TP$^+$}, which combines both {\em Hilbert} and {\em TP}. Specifically, given a microdata $\mathcal{T}$, {\em TP$^+$} first invokes {\em TP} to produce a partition of $\mathcal{T}$, and then applies {\em Hilbert} on the residue set $R$ (produced by {\em TP}) to reduce the number of stars in the $l$-diverse table. As discussed in Section~\ref{sec:algorithm-discuss}, such a hybrid algorithm also returns an $O(l \cdot d)$ solution for the star minimization problem.

\extraspacing \noindent {\bf Datasets } Following \cite{xt06-b,gkkm07}, we experiment with two datasets, SAL and OCC, obtained from the {\em American Community Survey} \cite{rsa04}. Both SAL and OCC contain 600k tuples, each capturing the information about a U.S. adult. Specifically, SAL has a sensitive attribute {\em Income}, and 6 QI attributes {\em Age}, {\em Gender}, {\em Race}, {\em Marital Status}, {\em Birth Place}, {\em Education}, {\em Work Class}. OCC contains the same QI attributes as in SAL, but has a different sensitive attribute {\em Occupation}. Table~\ref{tbl:exp-attribute} illustrates the domain size of each attribute.

\begin{table*}[t]
\centering
\begin{small}
\begin{tabular}{|@{\hspace{3pt}}c@{\hspace{3pt}}|@{\hspace{3pt}}c@{\hspace{3pt}}|@{\hspace{3pt}}c@{\hspace{3pt}}|@{\hspace{3pt}}c@{\hspace{3pt}}|@{\hspace{3pt}}c@{\hspace{3pt}}|@{\hspace{3pt}}c@{\hspace{3pt}}|@{\hspace{3pt}}c@{\hspace{3pt}}|@{\hspace{3pt}}c@{\hspace{3pt}}|@{\hspace{3pt}}c@{\hspace{3pt}}|@{\hspace{3pt}}c@{\hspace{3pt}}|}
\hline  & {\em Age} & {\em Gender} & {\em Race} & {\em Marital Status} & {\em Birth Place} & {\em Education} & {\em Work Class} & {\em Income} & {\em Occupation} \\ \hline
Size & 79 & 2 & 9 & 6 & 56 & 17 & 9 & 50 & 50\\
\hline
\end{tabular}
\end{small}
\tblcapup
\caption{Attribute domain sizes} \label{tbl:exp-attribute}
\vspace{-1mm}
\end{table*}

Based on SAL, we generate $7$ sets of microdata, SAL-$1$, SAL-$2$, ..., SAL-$7$. Each table in SAL-$d$ ($1 \le d \le 7$) is a projection of SAL on {\em Income} and $d$ QI attributes. As SAL has $7$ quasi-identifers, totally there are ${7 \choose d}$ microdata tables in SAL-$d$. Similarly, we also construct $7$ sets of microdata OCC-$d$ ($1 \le d \le 7$) from OCC.

\extraspacing \noindent {\bf Quality of generalizations } In the first set of experiments, we investigate the effect of $l$ on the quality of the generalization produced by each technique. In particular, for any given $l$, we employ each algorithm to generate $l$-diverse versions of the microdata in SAL-$4$ (OCC-$4$). Then, the performance of an algorithm is gauged by the average number of stars, in the $l$-diverse generalization it generates for the ${7 \choose 4} = 35$ microdata tables in SAL-$4$ (OCC-$4$).

\begin{figure}[t]
\centering
\begin{tabular}{cc}
\hspace{-5mm}\includegraphics[height=32mm]{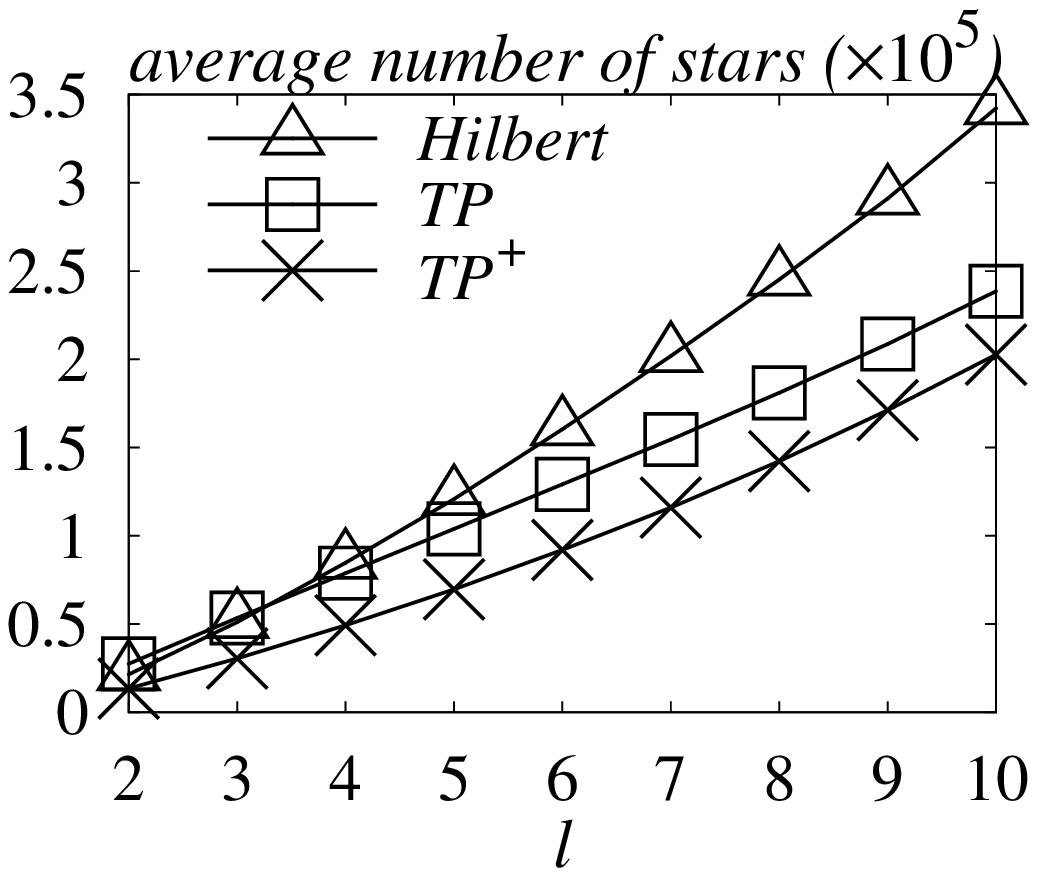} &
\hspace{-5mm}\includegraphics[height=32mm]{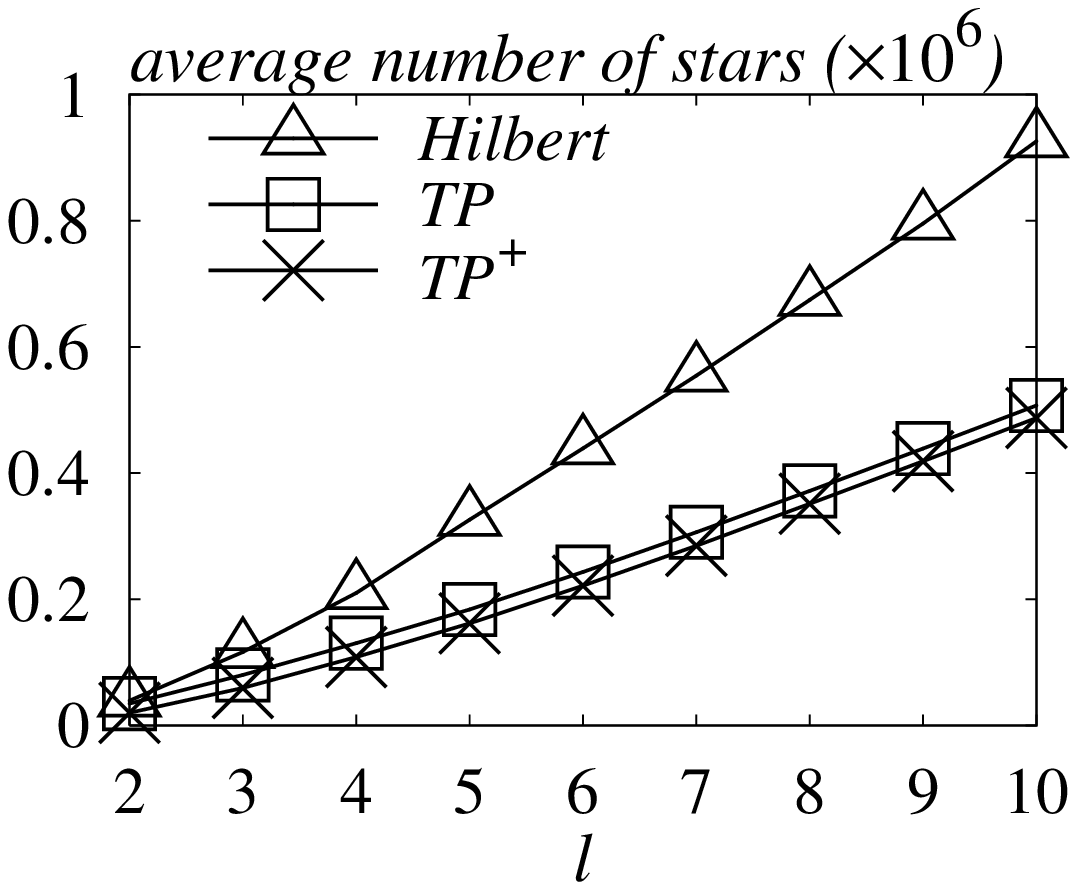} \\
\hspace{-0mm} (a) SAL-$4$ & \hspace{-5mm}(b) OCC-$4$
\end{tabular}
\figcapup \vspace{-2mm} \caption{Average number of stars vs.\ $l$}
\figcapdown
\label{fig:exp-stars-vs-l}
\end{figure}

Figure~\ref{fig:exp-stars-vs-l} illustrates the average number of stars as a function of $l$. All algorithms perform better when $l$ decreases, since a smaller $l$ leads to a lower degree of privacy protection, which can be achieved with less generalization. Both {\em TP} and {\em TP$^+$} consistently outperform {\em Hilbert}. In addition, {\em TP$^+$} incurs a smaller number of stars than {\em TP} in all cases.


\begin{figure}[t]
\centering
\begin{tabular}{cc}
\hspace{-5mm}\includegraphics[height=32mm]{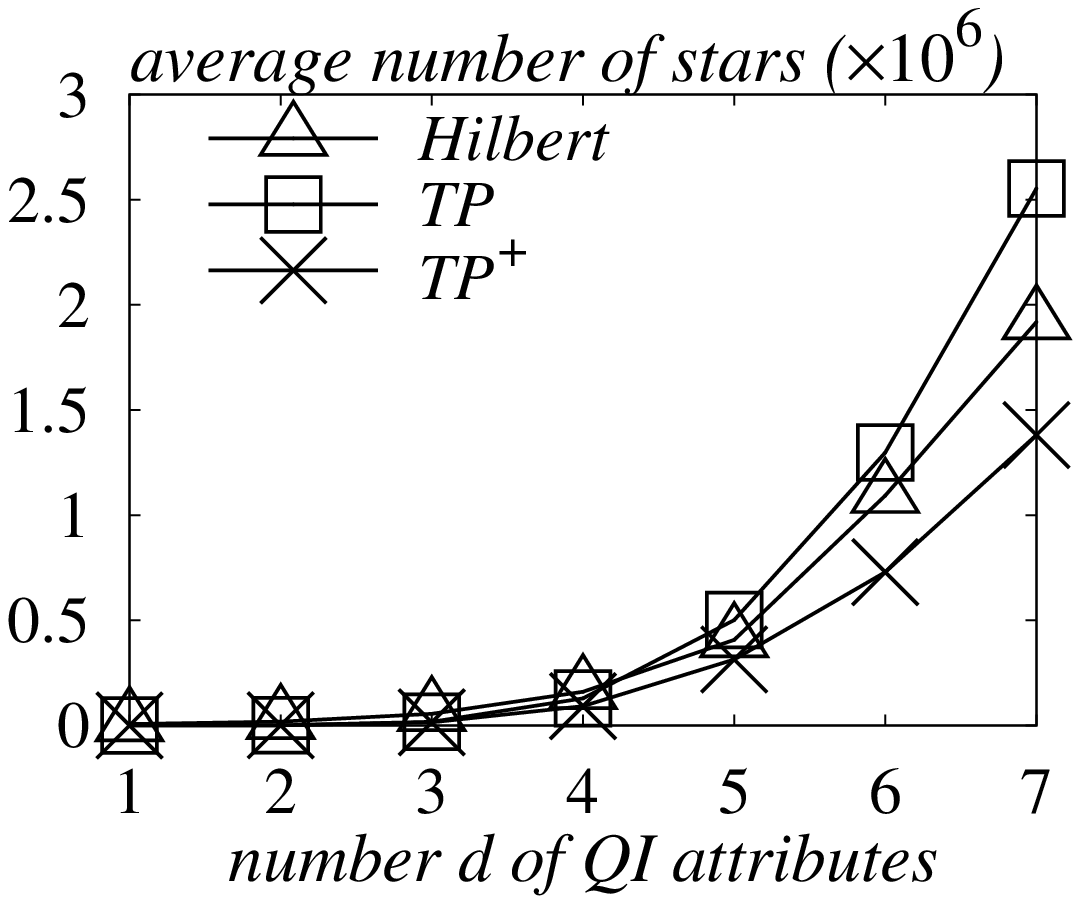} &
\hspace{-5mm}\includegraphics[height=32mm]{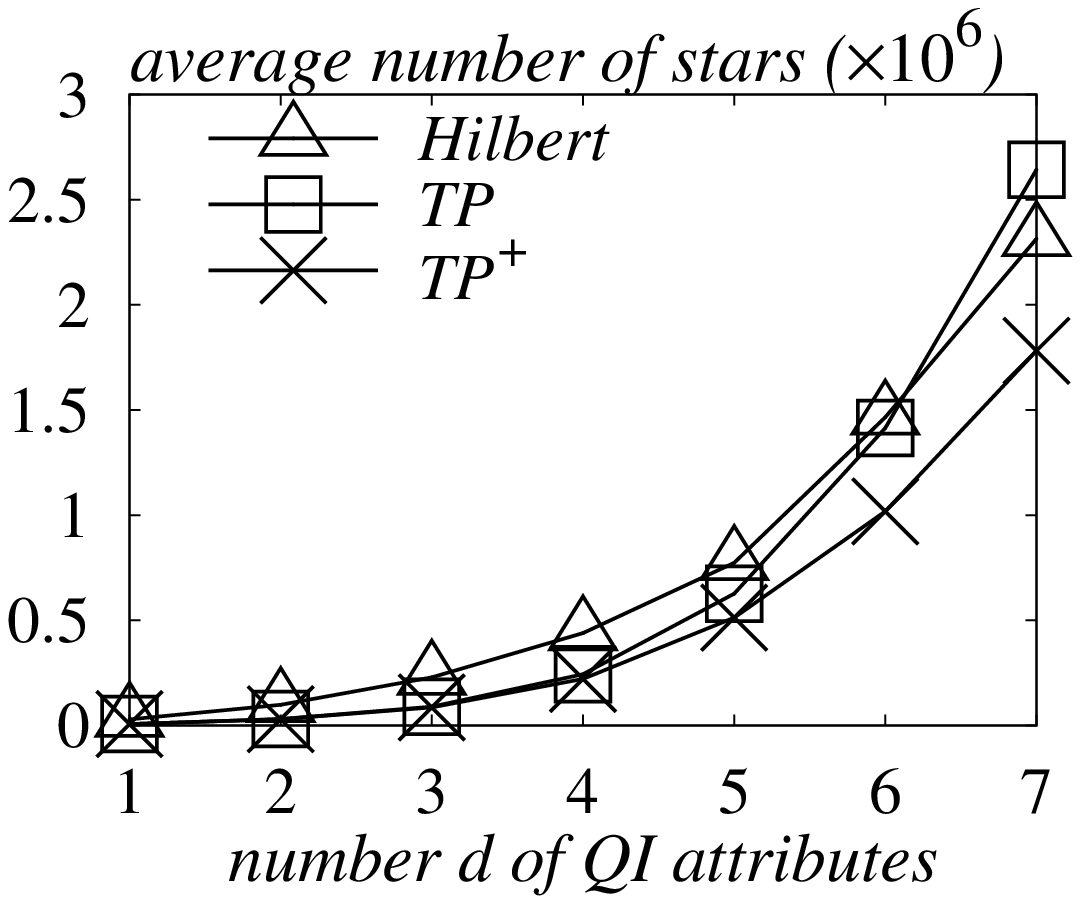} \\
\hspace{-0mm} (a) SAL-$d$ & \hspace{-5mm}(b) OCC-$d$
\end{tabular}
\figcapup \vspace{-2mm} \caption{Average number of stars vs.\ $d$ ($l = 6$)}
\label{fig:exp-stars-vs-d}
\end{figure}

Next, we examine the performance of each algorithm, fixing $l=6$ and varying the number $d$ of QI attributes in the microdata. Figure~\ref{fig:exp-stars-vs-d} shows the average number of stars incurred by each technique, for the tables in SAL-$d$ and OCC-$d$ ($1 \le d \le 7$). The average number of stars increases with $d$, which is consistent with the analysis in \cite{a05} that, all generalization techniques suffer from the curse of dimensionality. On SAL-$d$ (OCC-$d$), {\em TP} outperforms {\em Hilbert} when $d \le 4$ ($d \le 6$), but is inferior than {\em Hilbert} given a larger $d$. This is due to the fact that, as $d$ increases, the tuples in the microdata tend to have more diverse QI values, which, as discussed in Section~\ref{sec:algorithm-discuss}, renders {\em TP} less effective. {\em TP$^+$} overcomes this drawback by incorporating {\em Hilbert} to refine the residue set $R$, and hence, achieves better data utility than both {\em TP} and {\em Hilbert}.

\extraspacing \noindent {\bf Frequency of phase three execution } Recall that {\em TP} consists of three phases. For any positive integer $l$ and any microdata $\T$ with $d$ QI attributes, if {\em TP} terminates during the first or second phase, the number of stars in the returned generalization is at most $d \cdot (OPT + l - 1)$, where $OPT$ is the minimum number of stars in any $l$-diverse generalization of $\T$. In contrast, if {\em TP} terminates after phase three, the resulting generalization is an $(l \cdot d)$-approximation. Furthermore, the first two phases of {\em TP} have $O(n)$ time complexity, while the third phase runs in $O(s \cdot n)$ time in the worst case, where $s$ is the maximum number of tuples in $\T$ with distinct QI values. Therefore, {\em TP} performs much better in terms of both information loss and computation time, when it returns generalized tables without invoking phase three.

A natural question is, how often does {\em TP} execute the third phase? To answer this question, we apply {\em TP} on each microdata table in SAL-$d$ and OCC-$d$ ($1 \le d \le 7$) to compute its $l$-diverse ($2 \le l \le 10$) generalization, and examine whether {\em TP} invokes the third phase. It turns out that, on all 128 tables and for all 9 values of $l$, {\em TP} terminates before the third phase. In other words, in all our experiments, {\em TP} (and thus, {\em TP$^+$}) returns $O(d)$ solution to the star minimization problem.

\begin{figure}[t]
\centering
\begin{tabular}{cc}
\hspace{-5mm}\includegraphics[height=32mm]{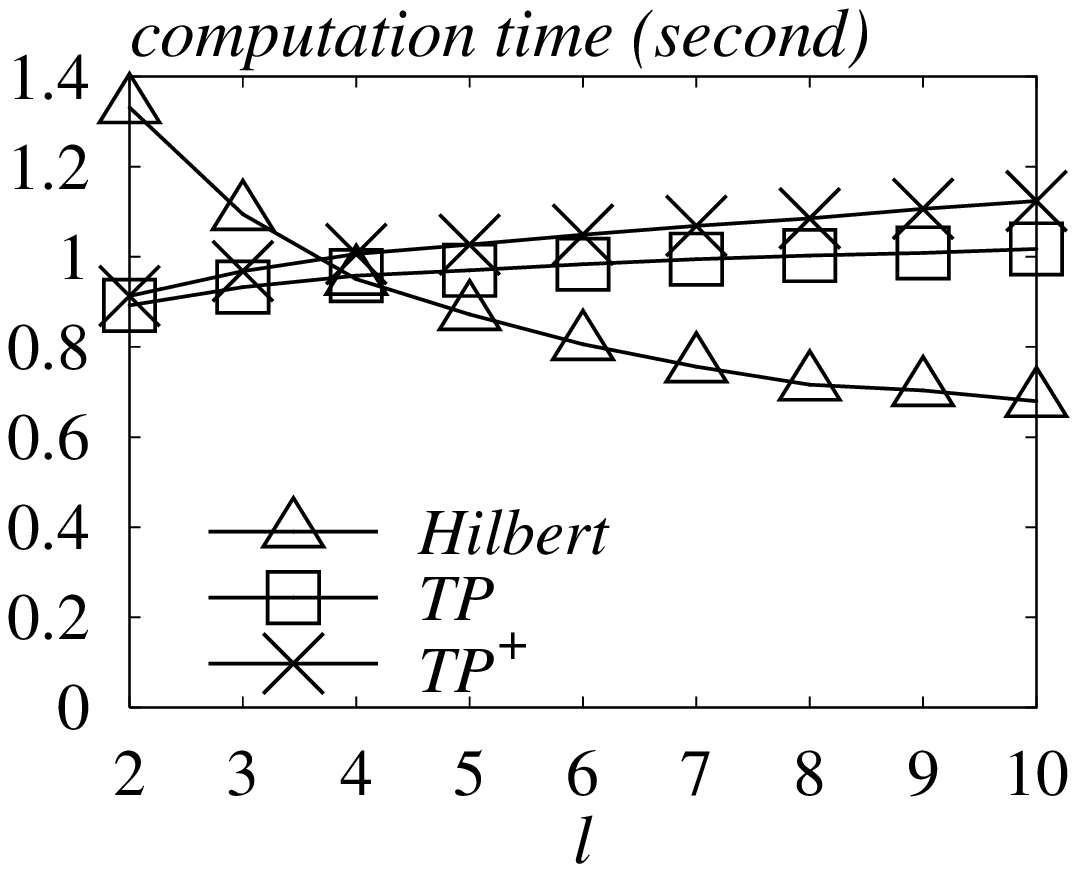} &
\hspace{-5mm}\includegraphics[height=32mm]{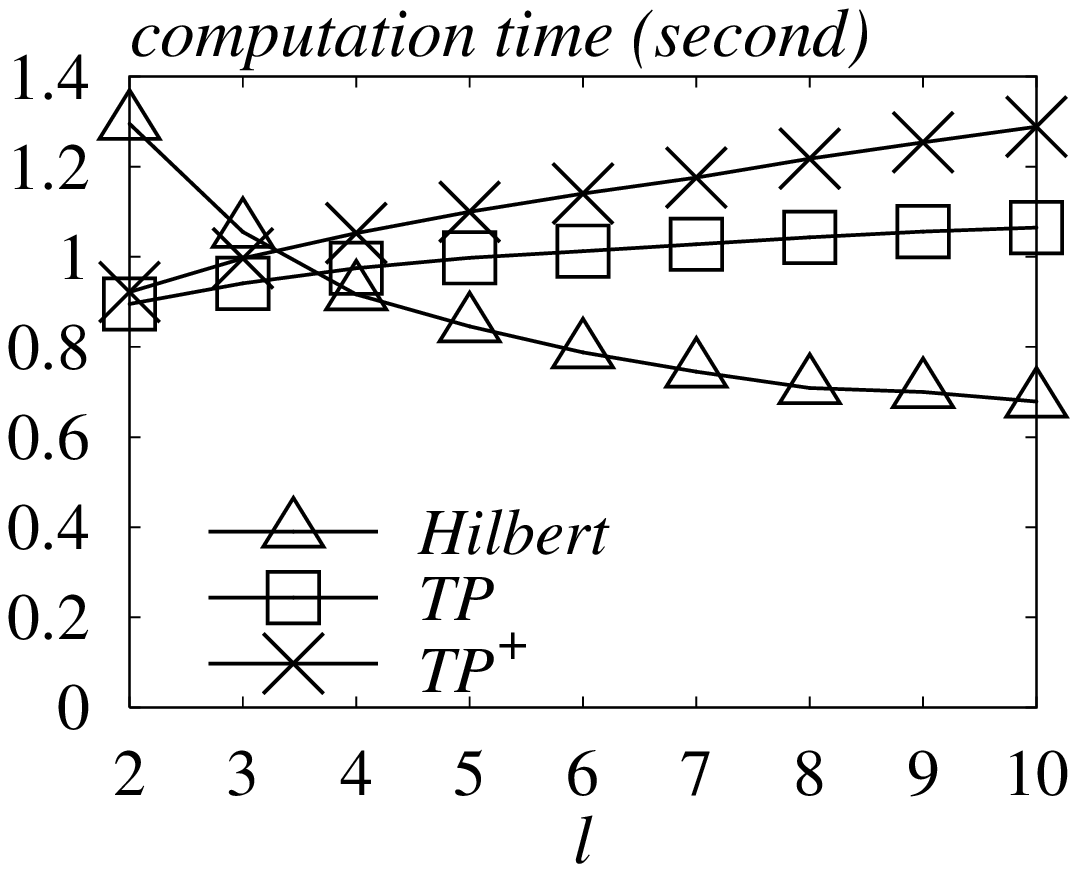} \\
\hspace{-0mm} (a) SAL-$4$ & \hspace{-5mm}(b) OCC-$4$
\end{tabular}
\vspace{-4mm} \caption{Computation time vs.\ $l$}
\vspace{-4mm} \label{fig:time-vs-l}
\end{figure}

\extraspacing \noindent {\bf Computation overhead } In the following experiments, we compare the efficiency of each algorithm. First, for any $l \in [2, 10]$, we examine the average time required by each technique to generate $l$-diverse versions of the microdata in SAL-$d$ (OCC-$d$). Figure~\ref{fig:time-vs-l} illustrates the computation time as a function of $l$. The overhead of {\em Hilbert} decreases with the increase of $l$, which is also observed in \cite{gkkm07}. In contrast, the computation cost of {\em TP} and {\em TP$^+$} increases with $l$. To understand this, recall that {\em TP} works by first dividing the tuples into QI-groups, and then iteratively moving tuples from each QI-group to the residue set $R$, until all QI-groups and $R$ become $l$-eligible. Given a larger $l$, {\em TP} has to remove more tuples from each QI-group to achieve $l$-eligibility, resulting in higher computation cost. In turn, this indicates that the residue set $R$ becomes larger, when $l$ increases. Consequently, the running time of {\em TP$^+$} also increases with $l$, because {\em TP$^+$} post-processes the output of {\em TP} by invoking {\em Hilbert} on $R$, the cost of which increases with the size of $R$.

\begin{figure}[t]
\centering
\begin{tabular}{cc}
\hspace{-5mm}\includegraphics[height=32mm]{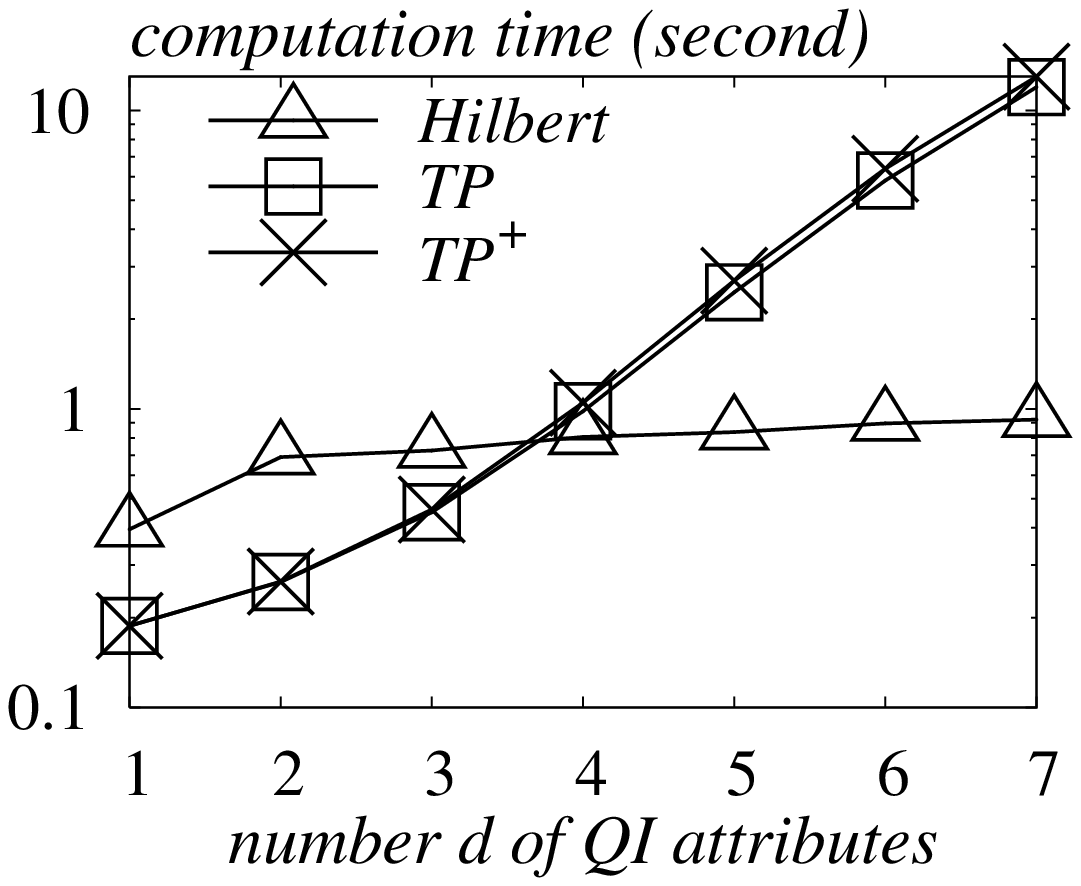} &
\hspace{-5mm}\includegraphics[height=32mm]{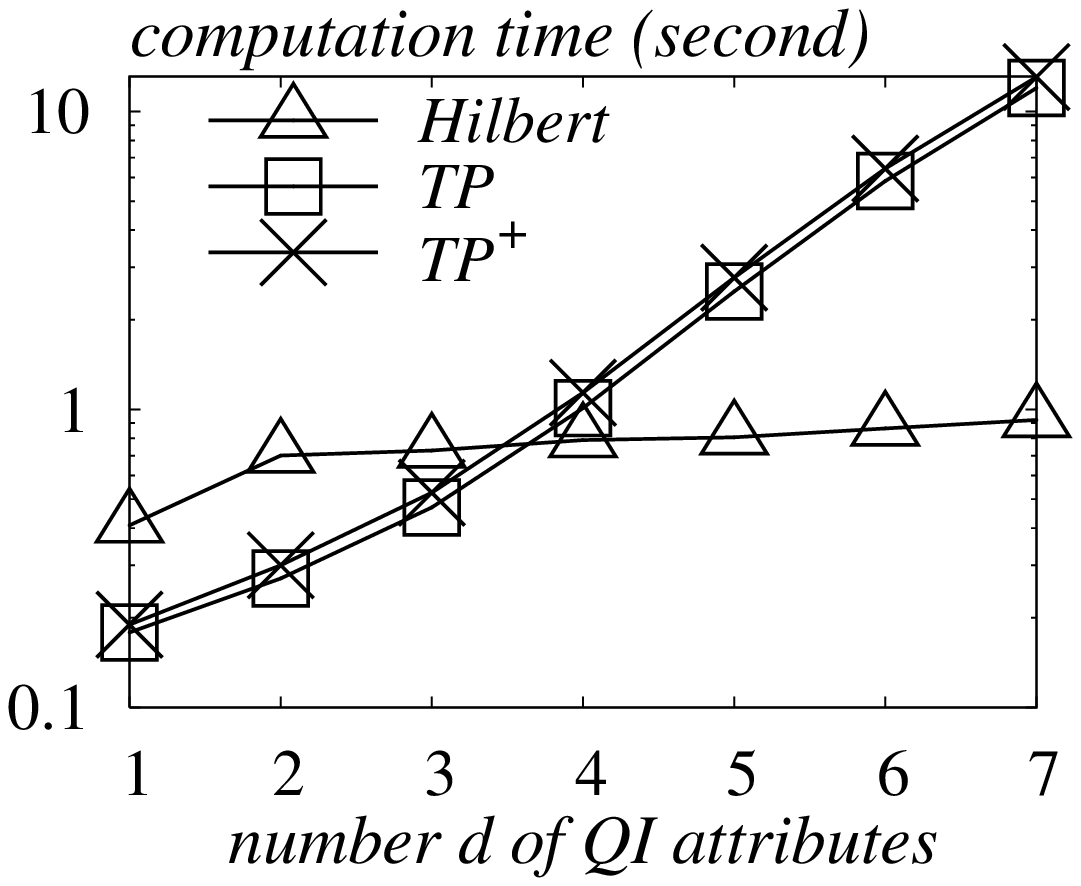} \\
\hspace{-0mm} (a) SAL-$d$ & \hspace{-5mm}(b) OCC-$d$
\end{tabular}
\vspace{-4mm} \caption{Computation time vs.\ $d$ ($l = 4$)}
\vspace{-0mm} \label{fig:time-vs-d}
\end{figure}

Next, we fix $l=6$, and investigate the average computation time of each algorithm on the microdata in AGE-$d$ (OCC-$d$), varying $d$ from $1$ to $7$. Figure~\ref{fig:time-vs-d} illustrates the results. The computation cost of {\em TP} increases with $d$. This is because, when $d$ is large, {\em TP} has to employ more generalization on the microdata to achieve $l$-diversity (see Figure~\ref{fig:exp-stars-vs-d}). As a result, {\em TP} needs to move a larger number of tuples from the QI-groups to the set $R$, leading to higher processing overhead. Because {\em TP$^+$} incorporates {\em TP}, its computation time also increases with $d$. The efficiency of {\em Hilbert} is insensitive to $d$, which is consistent with the experimental results in \cite{gkkm07}.

\begin{figure}[t]
\centering
\begin{tabular}{cc}
\hspace{-5mm}\includegraphics[height=32mm]{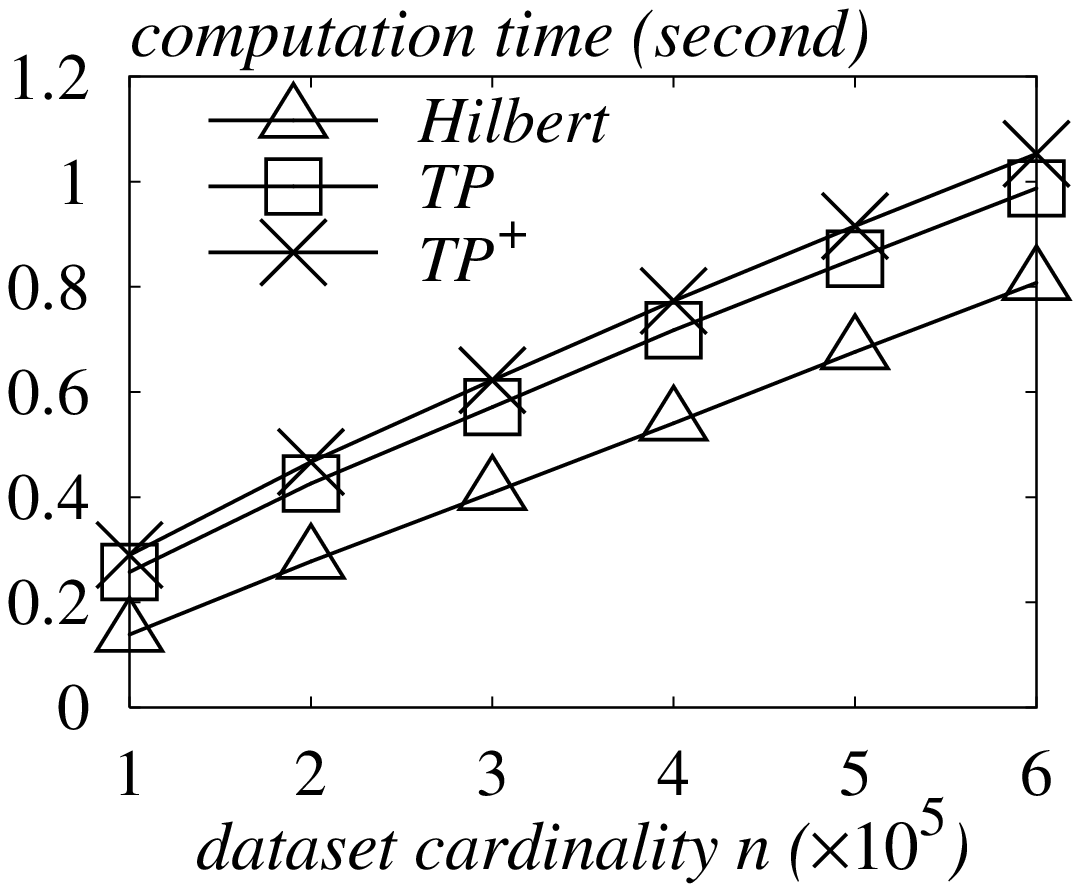} &
\hspace{-5mm}\includegraphics[height=32mm]{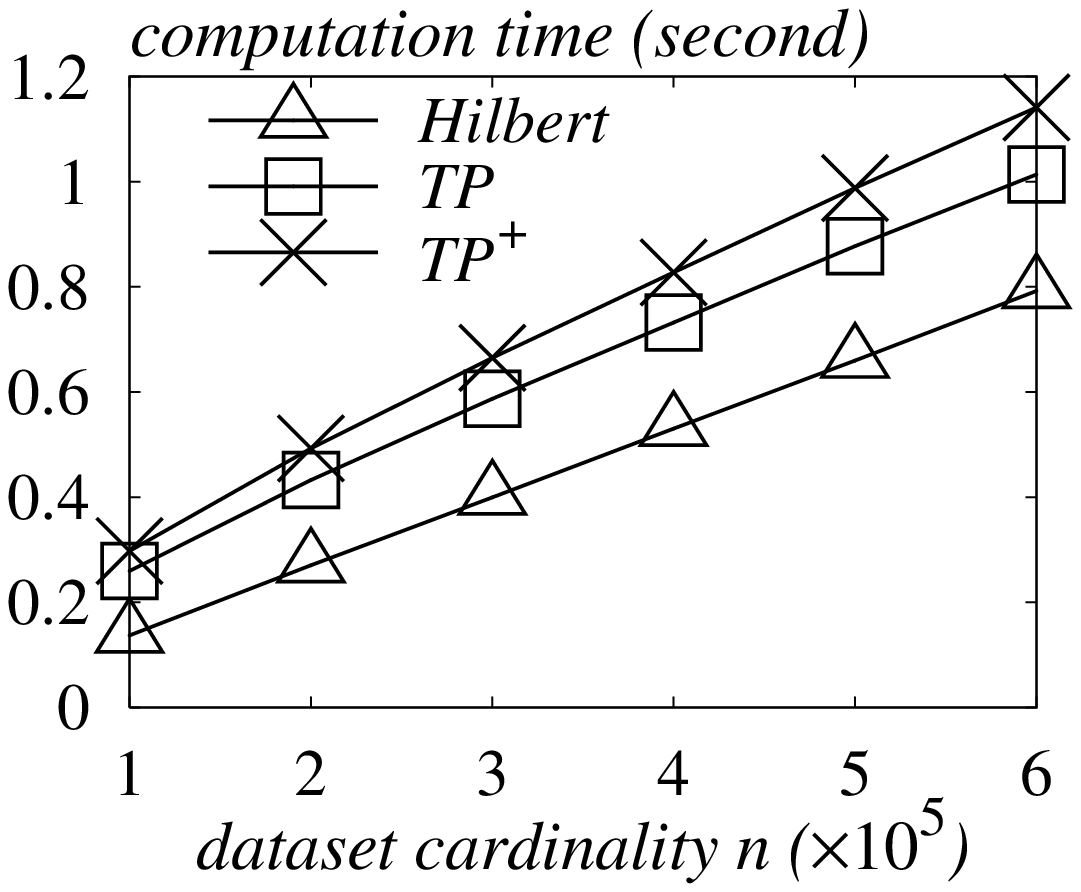} \\
\hspace{-0mm} (a) SAL-$4$ & \hspace{-0mm}(b) OCC-$4$
\end{tabular}
\vspace{-4mm} \caption{Computation time vs.\ $n$ ($l=6$)}
\vspace{-4mm} \label{fig:time-vs-n}
\end{figure}

Finally, we study the effect of dataset cardinality $n$ on the computation time of generalization. For each table $\mathcal{T}$ in SAL-$4$ and OCC-$4$, we generate various sample sets of $\mathcal{T}$, with sample size varying from $100$k to $600$k. After that, we employ each algorithm to compute a $6$-diverse generalization of each sample set, and measure the average running time of the algorithm. Figure~\ref{fig:time-vs-n} plots the computation overhead as a function of the dataset cardinality $n$. The running time of each technique is less than $1.2$ seconds even for the largest datasets. The processing cost of {\em TP} increases linearly with $n$. This is expected, since (i) {\em TP} bypasses phase three in all cases, and (ii) the first and second phases of {\em TP} have linear time complexity. The computation time of {\em Hilbert} is almost linear, which confirms the analysis in \cite{gkkm07} that {\em Hilbert} runs in $O(n \log n)$ time. Since both {\em TP} and {\em Hilbert} scale well with $n$, {\em TP$^+$} (as a combination of {\em TP} and {\em Hilbert}) also achieves satisfactory scalability.



\extraspacing \noindent {\bf Summary } In terms of data utility, {\em TP$^+$} significantly outperforms not only {\em TP} but also {\em Hilbert}, the best existing algorithm that can achieve $l$-diversity via suppression. In terms of computation time, {\em Hilbert} is superior than {\em TP} and {\em TP$^+$}. Nevertheless, as the anonymization of microdata incurs only one-time cost,  computational efficiency is not a major concern in data publishing. This makes {\em TP$^+$} more preferable than {\em Hilbert} for suppression-based anonymization.


\subsection{Comparison with Single-Dimensional Generalization} \label{sec:exp-kl}

Having established {\em TP$^+$} as an excellent suppression-based algorithm, in this section we will move on to compare {\em TP$^+$} with the single- and multi-dimensional generalization methods. First, we observe that multi-dimensional generalization always guarantees higher data utility than suppression. Specifically, given any table $\mathcal{T}^*$ generated by suppression, we may transform it into a multi-dimensional generalization $\mathcal{T}^{*\prime}$, by replacing each star on a QI attribute $A$ with a sub-domain of $A$, such that the sub-domain contains all $A$ values appearing in the QI-group. As each sub-domain captures more accurate information than a star, $\mathcal{T}^{*\prime}$ always incurs less information loss than $\mathcal{T}^*$. For example, let us consider Table~\ref{tbl:intro-micro-2-div}, which contains four stars on {\em Age} and {\em Education}, respectively, and all the stars appear in the first QI-group. We may replace each star on {\em Age} with a sub-domain ``$<$50'', as it covers the {\em Age} values of all tuples in the QI-group (see Table~\ref{tbl:intro-micro}). Similarly, each star on {\em Education} can be replaced with a sub-domain ``Bachelor or above''. This results in the multi-dimensional generalization in Table~\ref{tbl:related-multi-dimen}, which apparently contains more accurate information than Table~\ref{tbl:intro-micro-2-anony}.

As discussed in Section~\ref{sec:related}, however, multi-dimensional generalization produces anonymized data that is unusable by off-the-shelf statistical package, whereas suppression does not suffer from this drawback. Consequently, even though multi-dimensional generalization outperforms suppression in terms of data utility, it cannot be chosen over suppression in the scenarios where software support for anonymized data is a concern. Yet, in such scenarios, suppression is not the only applicable anonymization method, as single-dimensional generalization can also generate data that can be directly fed into commercial statistical software. This leads to an interesting question: how does {\em TP$^+$} compare to the existing single-dimensional generalization methods in terms of data utility?

To answer the above question, we implement {\em TDS}\footnote{{\em TDS} was initially designed for $k$-anonymity. We modify it into an $l$-diversity algorithm to facilitate the comparison with {\em TP$^+$}.}, the state-of-the-art single-dimensional generalization algorithm proposed in \cite{fwy05}, and compare it against {\em TP$^+$} on the quality of generalization. Following \cite{kg06,gkkm07}, we measure the quality of a generalized table $\mathcal{T}^*$, by the similarity between the multi-dimensional distribution induced by $\mathcal{T}^*$ and the distribution induced by the microdata $\mathcal{T}$. To explain this, observe that each tuple in $\mathcal{T}$ can be regarded as a point in a $(d+1)$-dimensional space $\Omega$, where the $i$-th ($1 \le i \le d$) dimension corresponds to the $i$-th QI attribute in $\mathcal{T}$, and the $(d+1)$-th dimensional corresponds to the sensitive attribute. As such, $\mathcal{T}$ can be captured by a probabilistic density function (pdf) $f$ defined on $\Omega$, such that, for any point $p \in \Omega$, $f(p)$ equals the fraction of tuples in $\mathcal{T}$ represented by $p$.

Similarly, any generalization $\mathcal{T}^*$ of $\mathcal{T}$ defines a pdf $f^*$ on $\Omega$. In particular, if a tuple $t^* \in \mathcal{T}^*$ has a star on an attribute $A$, we treat $t^*[A]$ as a random variable uniformly distributed in the domain of $A$; on the other hand, if $t^*[A]$ is a sub-domain of $A$, we treat $t^*[A]$ as uniformly distributed in the sub-domain. As in \cite{kg06,gkkm07}, we gauge the similarity between $f$ and $f^*$ by their {\em KL-divergence} \cite{kl51}, defined as
\begin{equation} \label{eq:exp-kl}
KL(f, f^*) = \sum_{p \in \Omega} f(p) \cdot \ln\frac{f(p)}{f^*(p)}.
\end{equation}
A smaller $KL(f, f^*)$ indicates a higher degree of similarity between $f$ and $f^*$.

\begin{figure}[t]
\centering
\begin{tabular}{cc}
\hspace{-5mm}\includegraphics[height=32mm]{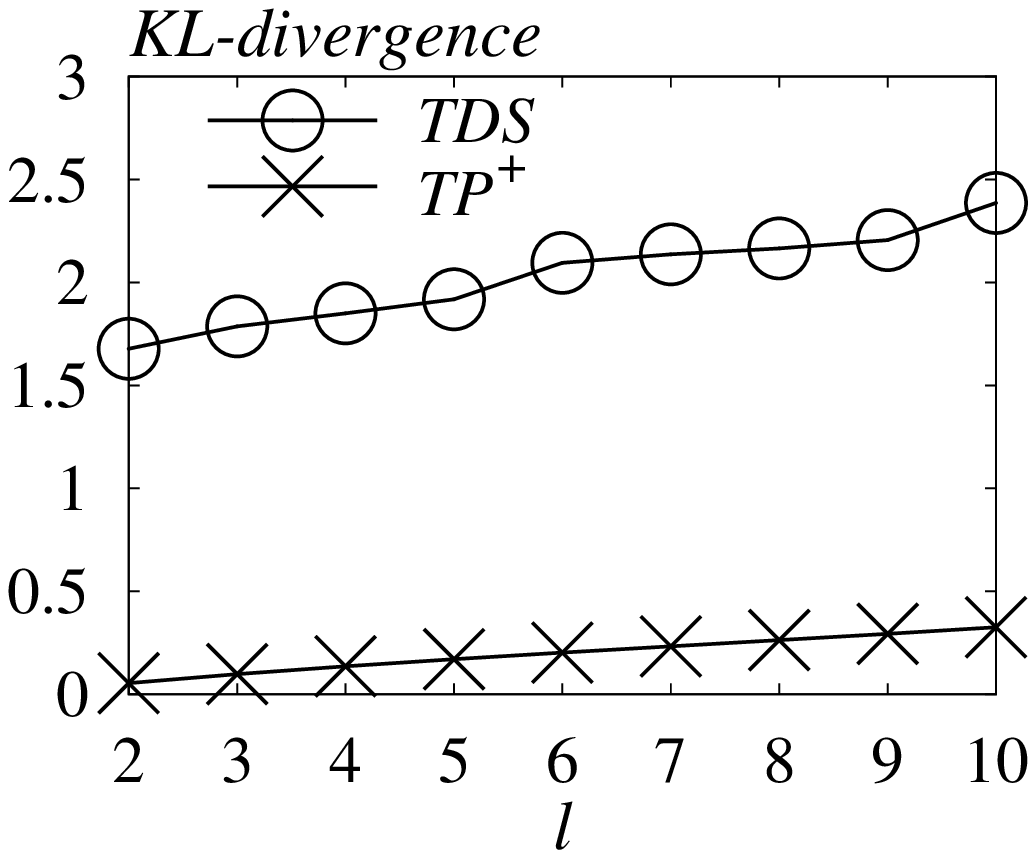} &
\hspace{-5mm}\includegraphics[height=32mm]{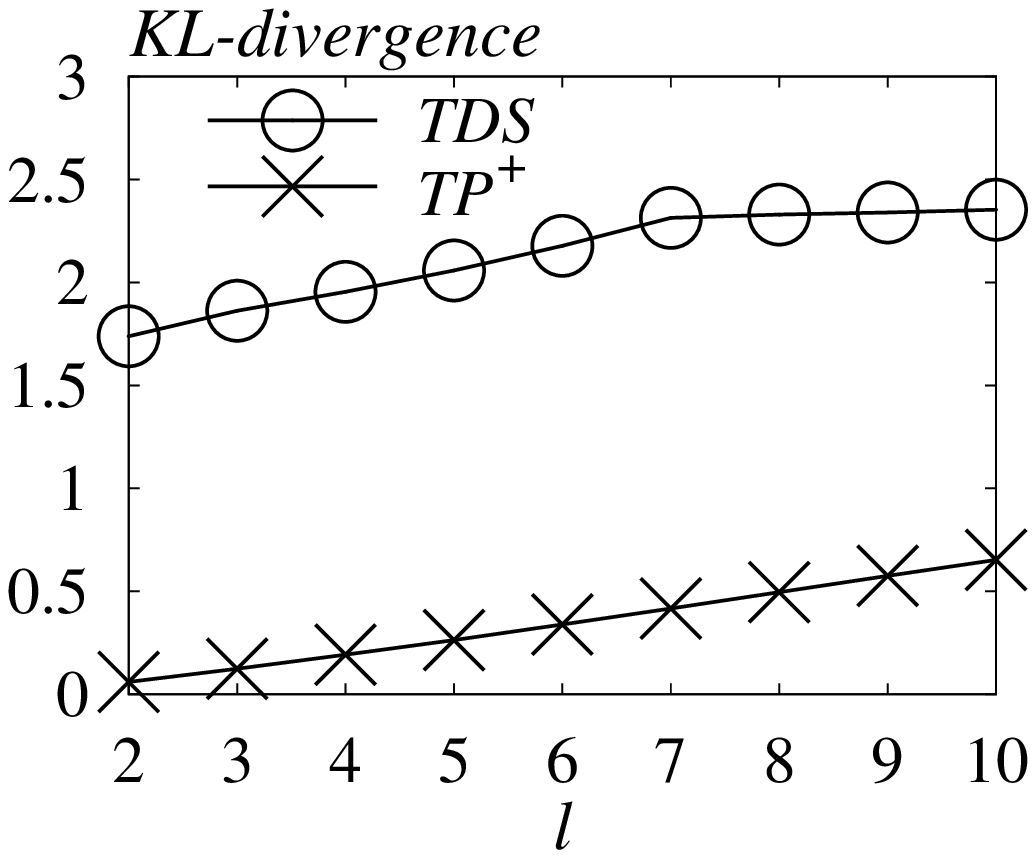} \\
\hspace{-0mm} (a) SAL-$4$ & \hspace{-5mm}(b) OCC-$4$
\end{tabular}
\vspace{-4mm} \caption{KL-divergence vs.\ $l$}
\vspace{-0mm} \label{fig:exp-kl-vs-l}
\end{figure}

In the first set of experiments, we apply {\em TP$^+$} and {\em TDS} on the microdata in SAL-$4$ and OCC-$4$, varying $l$ from $2$ to $10$. Figure~\ref{fig:exp-kl-vs-l} plots the average KL-divergence incurred by each algorithm as a function of $l$. {\em TP$^+$} significantly outperforms {\em TDS} in all cases. The KL-divergence entailed by {\em TP$^+$} increases with $l$, which is consistent with the results in Figure~\ref{fig:exp-stars-vs-l} that, a larger $l$ leads to more stars in the generalized table.

\begin{figure}[t]
\centering
\begin{tabular}{cc}
\hspace{-5mm}\includegraphics[height=32mm]{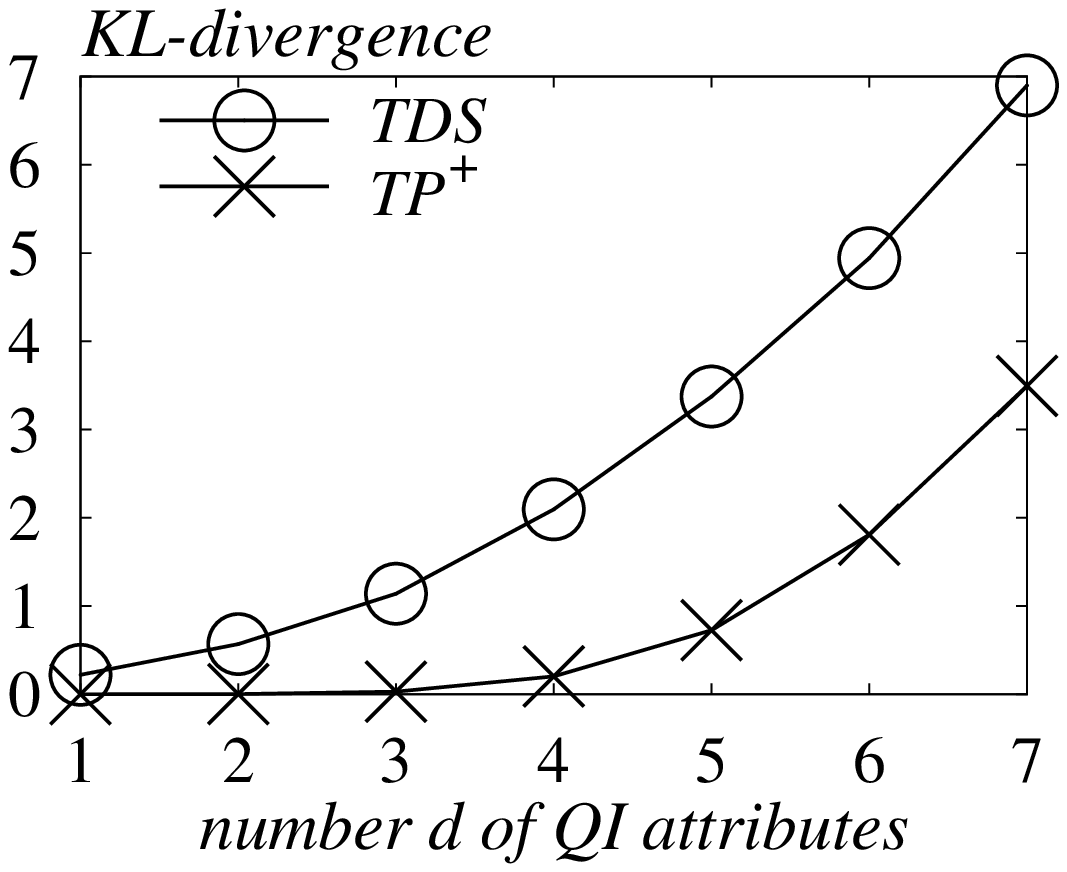} &
\hspace{-5mm}\includegraphics[height=32mm]{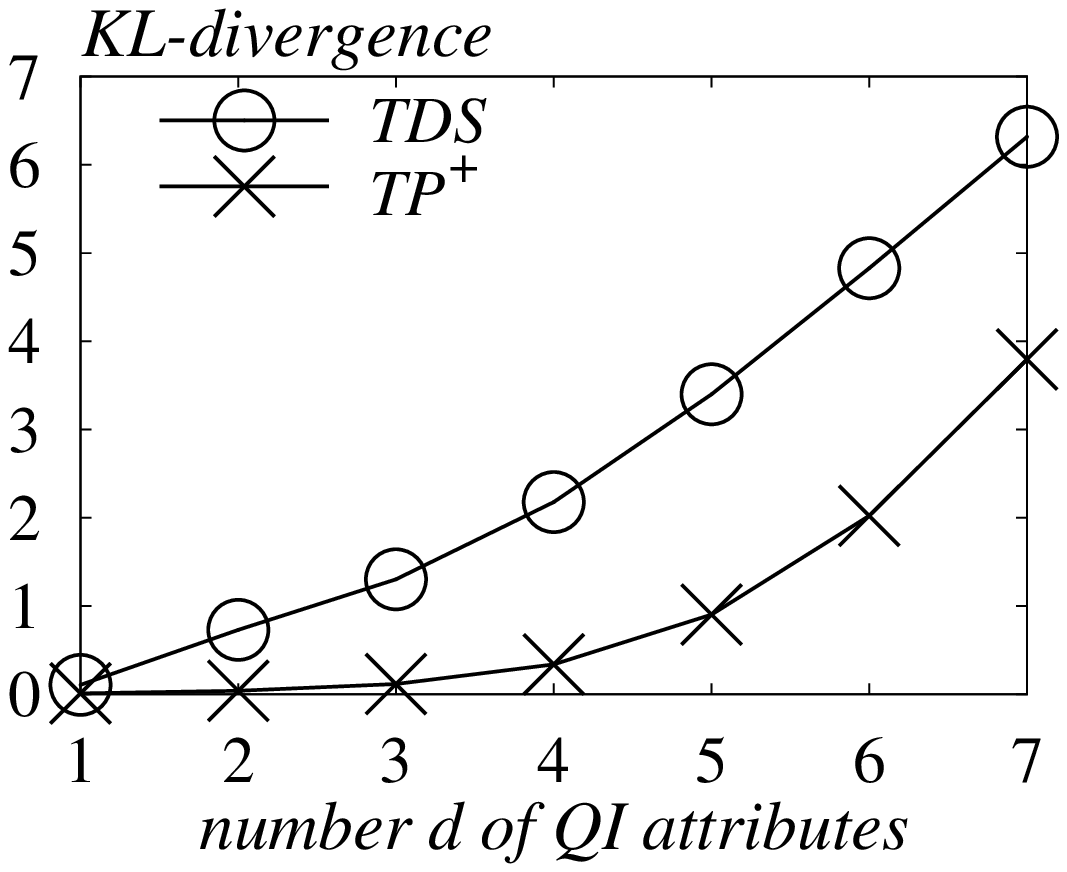} \\
\hspace{-0mm} (a) SAL-$d$ & \hspace{-5mm}(b) OCC-$d$
\end{tabular}
\vspace{-4mm} \caption{KL-divergence vs.\ $d$ ($l = 6$)}
\vspace{-4mm} \label{fig:exp-kl-vs-d}
\end{figure}

Next, we fix $l=6$, and measure the average KL-divergence incurred by {\em TP$^+$} and {\em TDS} in anonymizing the microdata in SAL-$d$ (OCC-$d$). Figure~\ref{fig:exp-kl-vs-d} illustrates the average KL-divergence as a function of $d$. Again, the information loss caused by {\em TP$^+$} is consistently smaller than {\em TDS}. The performance of both algorithms degrades with the increase of $d$, since, as mentioned in Section~\ref{sec:exp-star}, all generalization methods inevitably suffer from the curse of dimensionality.

In summary, {\em TP$^+$} achieves significantly higher data utility than {\em TDS}. This makes {\em TP$^+$} a favorable choice for data publishers who aim to release generalized tables that can be easily analyzed using existing statistical software. Multi-dimensional generalization methods, on the other hand, should be adopted when the users are equipped with their own tools for analyzing complex anonymized data.

\section{Conclusions} \label{sec:conclude}

The existing work on $l$-diversity focuses on the development of heuristic solutions. In this paper, we present the first theoretical study on the complexity and approximation algorithms of $l$-diversity. First, we prove that computing the optimal $l$-diverse generalization is NP-hard, for any $l \ge 3$. After that, we develop an $O(l \cdot d)$-approximation algorithm for the problem, where $d$ denotes the number of QI attributes in the microdata. The effectiveness and efficiency of the proposed technique are verified through extensive experiments.

There exist several promising directions for future work. First, we plan to improve our three phase algorithm, to achieve a better approximation ratio for the star minimization problem. Second, we have only considered categorical domains in this paper, in the future we will try to extend our algorithm to support numerical domains. Finally, it is interesting to investigate the hardness and approximation algorithms for other privacy principles.

%

\begin{small}
\bibliographystyle{abbrv}
\bibliography{ref}
\end{small}

\pagebreak

\section{Appendix} \label{sec:appendix}

\input{proofs}

\end{sloppy}
\end{document}

%% file: proofs.tex
\noindent {\bf Proof of Lemma~\ref{lem:phase-one} }
Consider the following two cases.  If initially $h(Q_i, v) \le
h(\Qo_i)$, then the algorithm must have not removed any tuple of SA
value $v$; hence $h(Q'_i,v) \le h(\Qo_i, v)$ trivially.  If $h(Q_i,
v) > h(\Qo_i)$ initially, then the algorithm will reduce the number
of tuples in $Q_i$ to exactly $h(\Qo_i)$, i.e., $h(\Qo_i, v) =
h(\Qo_i)$.  We argue that $h(Q'_i) \le h(\Qo_i)$, which implies
$h(Q'_i, v) \le h(Q'_i) \le h(\Qo_i) = h(\Qo_i, v)$.  Assume by
contradiction that $h(Q'_i) > h(\Qo_i)$.  Consider the set $Q''_i$
such that $h(Q''_i, v) = \min\{h(Q'_i), h(Q_i,v)\}$ for all $v$,
i.e., $Q''_i$ is obtained from $Q_i$ by reducing the number of
tuples of each SA value $v$ to $h(Q'_i)$, whenever
$h(Q_i, v) > h(Q'_i)$. Note that we must have $Q'_i \subseteq
Q''_i$. Since $Q'_i$ is $l$-eligible and $h(Q'_i) = h(Q''_i)$,
$Q''_i$ must also be $l$-eligible. Thus, the algorithm would have
stopped earlier at $Q''_i$, whose pillar height is higher than that
of $\Qo_i$, reaching a contradiction. \done

\extraspacing \noindent
{\bf Proof of Corollary~\ref{cor:phaseone} }
  Let $Q'_1, \dots, Q'_s, R'$ be an optimal solution.  Since $Q'_i$ is
  $l$-eligible, by Lemma~\ref{lem:phase-one}, we have $h(Q'_i, v) \le
  h(\Qo_i, v)$.  Summing over all $v$, we have $|Q'_i| \le |\Qo_i|$, and
  thus $|R'| = n - \sum_{i=1}^s |Q'_i| \ge n - \sum_{i=1}^s |\Qo_i| =
  |\Ro|$.  If the algorithm stops after phase one, $\Ro$ is also
  $l$-eligible, and $\Qo_1, \dots, \Qo_s, \Ro$ is a valid solution;
  hence $|\Ro| = |R'|$. \done

\extraspacing \noindent
{\bf Proof of Corollary~\ref{cor:optlower} }
  Let $Q'_1, \dots, Q'_s, R'$ be an optimal solution. By
  Lemma~\ref{lem:phase-one}, $h(Q'_i, v) \le h(\Qo_i, v)$ for any $v$.
  Summing over all $i$, we have $\sum_{i=1}^s h(Q'_i, v) \le \sum_{i=1}^s
  h(\Qo_i, v)$.  Since $h(R',v) + \sum_{i=1}^s h(Q'_i, v) = h(\Ro,v) +
  \sum_{i=1}^s h(\Qo_i, v)$, we have $h(R',v) \ge h(\Ro,v)$, in particular,
  $h(R') \ge h(\Ro)$. Since $R'$ is $l$-eligible, $OPT = |R'| \ge l \cdot
  h(R') \ge l \cdot h(\Ro)$. \done

\extraspacing \noindent
{\bf Proof of Lemma~\ref{lmm:rt=ro} } The lemma is equivalent to the claim that phase two
never picks a pillar of $R$ to move tuples to. Indeed, if a pillar
$p$ of $R$ is picked in some iteration, then there must be an alive
QI-group $Q$ such that $h(Q,p) > 0$. Since $Q$ is $l$-eligible, it
contains at least $l$ different SA values, i.e., $h(Q,v) > 0$ for
each of those values $v$. As $Q$ is alive, by definition, all the SA
values in $Q$ are alive. On the other hand, $R$ has at most $l-1$
pillars; otherwise, $R$ is $l$-diverse, and the current iteration
should not have started. Hence, there should be at least one alive
SA value that is not a pillar in $R$. So the algorithm should have
picked that value instead of $p$ (recall that each iteration selects
the least frequent alive SA value in $R$). \done

\extraspacing \noindent
{\bf Proof of Lemma~\ref{lem:phasetwo} } Since the algorithm did not stop after phase one, $\Ro$ is not
$l$-eligible, implying $|\Ro| < l \cdot h(\Ro)$.  As $h(\Ro) =
h(\Rt)$ (Lemma~\ref{lmm:rt=ro}), the algorithm will stop as soon as
$|R|$ reaches $l \cdot h(\Ro)$.  In each iteration
of phase two, we remove at most $l$ tuples together, since a thin
$l$-eligible QI-group has at most $l$ pillars.  Therefore, $|R|$ at
most exceeds $l\cdot h(\Ro)$ by $l-1$ when the algorithm terminates.
\done

\extraspacing \noindent
{\bf Proof of Lemma~\ref{lem:Rt} }
  Assume for contradiction that $p$ is a conflicting pillar in all $Q_i$.
  Since $R$ is not $l$-eligible, we have
\begin{equation}
\label{eq:4}
  |R| < l \cdot h(R, p).
\end{equation}
For any $i$, since $Q_i$ is thin and has $p$ as one of its
conflicting pillars, we have
\begin{equation}
  \label{eq:5}
  |Q_i| = l \cdot h(Q_i, p).
\end{equation}
Summing (\ref{eq:5}) over all $i$ and (\ref{eq:4}), we have
\[ n < l \cdot h(\T, p), \]
where $h(\T, p)$ represents the total number of tuples with SA value
$p$ in the microdata $\T$.  This contradicts with the assumption
that $\T$ is $l$-eligible. \done

\extraspacing \noindent
{\bf Proof of Corollary~\ref{cor:Rt} }
  If $\Rt$ has only one pillar, then all $\Qt_i$ can only conflict with
  $\Rt$ on this pillar, contradicting Lemma~\ref{lem:Rt}. \done

\extraspacing \noindent
{\bf Proof of Theorem~\ref{thm:l2} } 
If the algorithm terminates in phase one, then the theorem follows from
Corollary~\ref{cor:phaseone}.  Otherwise, it must terminate during phase
two, due to Corollary~\ref{cor:Rt} and the fact that if $R$ has at least two
pillars, then it must be $2$-eligible.  Then the theorem follows from
Corollary~\ref{cor:twophase-lower}. \done

\extraspacing \noindent
{\bf Proof of Lemma~\ref{lem:round1} }
  Let $Q_1, \dots, Q_s, R$ be the status at the beginning of a particular
  round.  We know that all the $Q_i$'s are dead, and $|R| < l \cdot h(R)$.
  Thus Lemma~\ref{lem:Rt} still holds on $Q_1, \dots, Q_s, R$, i.e., for
  any pillar $p$ of $R$, there exists a QI-group in which $p$ is not a
  conflicting pillar.  In other words, $p$ is covered by at least one
  $\overline{C(Q)}$.  Thus, the greedy algorithm will pick at most $l-1$
  QI-groups before it finishes ($R$ has at most $l-1$ pillars).  Afterward,
  each pillar of these QI-groups will ship a tuple to $R$.  We distinguish
  between two cases.  For a pillar $p$ of $R$, $h(R,p)$ increases by at
  most $l-2$ since there is at least one QI-group in which $p$ is not a
  pillar.  For other SA values $v$ of $R$, $h(R,v)$ increases by at most
  $l-1$.  But since $h(R,v) \le h(R) -1$, these other SA values will not
  cause $h(R)$ to increase by more than $l-2$, either. \done

\extraspacing \noindent
{\bf Proof of Lemma~\ref{lem:round2} }
  Define the {\em gap} for $R$ to reach $l$-eligibility as $\Delta(R)= l
  \cdot h(R) - |R|$.  The algorithm will terminate as soon as the gap
  reduces to zero or negative.  At the beginning of phase three, we have
\begin{equation}
\label{eq:2}
 \Delta(\Rt) = l \cdot h(\Rt) - |\Rt| \le  l \cdot h(\Rt).
\end{equation} 

Next we consider how much the gap reduces in each round.  Suppose in
the first step of a round, the greedy algorithm picks $r$ QI-groups.
Following the same reasoning as in the proof of
Lemma~\ref{lem:round1}, $h(R)$ increases by at most $r-1$.  On the
other hand, for any QI-group $Q$ picked by the greedy algorithm, its
pillar height $h(Q)$ decreases by one.  In the second step of this
round, we remove tuples from $Q$ until it becomes thin again,
meaning that a total of $l$ tuples (including those removed in the
first step) must have been removed in this round.  Henceforth, $|R|$
has increased by at least $l \cdot r$ tuples in this round.  So the
net effect is that $\Delta(R)$ must have decreased by at least $l
\cdot r - l(r-1) = l$ tuples.

Combining with (\ref{eq:2}), we conclude that the total number of
rounds is at most $\Delta(\Rt) / l \le h(\Rt)$. \done

\extraspacing \noindent
{\bf Proof of Theorem~\ref{thm:main} }
Let $\hat{R}$ be the final set of removed tuples at the end of phase
three.  By Lemmas~\ref{lem:round1} and ~\ref{lem:round2}, we have
\[ h(\hat{R}) \le h(\Rt) + (l-2) h(\Rt) = (l-1)h(\Rt). \]
Since in the second step of each round of phase three, we remove at
most $l$ tuples together and the algorithm terminates as soon as
$|R|$ reaches $l\cdot h(R)$, we have
\[ |\hat{R}| \le l \cdot h(\hat{R}) + l-1. \]
Also note that $h(\Ro) = h(\Rt)$, we have
\[ |\hat{R}| \le l(l-1)h(\Ro) + l - 1. \]
By Corollary \ref{cor:optlower}, we can bound the approximation
ratio as
\[ \frac{|\hat{R}|}{OPT} \le \frac{l(l-1)h(\Ro) + l -1}{l\cdot h(\Ro)} <
l; \] hence the proof. \done